\documentclass{tran-l}
\usepackage[square,sort,comma,numbers]{natbib}
\usepackage[T1]{fontenc}
\usepackage[utf8]{inputenc}
\setcitestyle{square}
\usepackage{amsmath,amssymb,amsfonts,amsthm}
\usepackage{mathtools}
\usepackage{mathrsfs}
\usepackage{chemfig}
\usepackage{chemformula}
\usepackage{centernot}
\usepackage{booktabs}
\usepackage{enumitem}  
\usepackage{url}
\usepackage{xcolor}
\definecolor{navy}{RGB}{0,0,128}
\definecolor{royalblue}{RGB}{65,105,225}
\usepackage{hyperref}
\hypersetup{
    linktocpage,
    colorlinks,
    citecolor=royalblue,
    filecolor=black,
    linkcolor=royalblue,
    urlcolor=black
}
\usepackage[left=3cm, right=3cm]{geometry}
\usepackage{braket}
\usepackage{physics}
\usepackage{amscd}
\usepackage{tikz-cd} 
\usepackage{tikz-3dplot}

\usepackage{algorithm}
\usepackage{algpseudocode}
\usetikzlibrary{arrows.meta,positioning,shapes.geometric,fit,backgrounds,calc}
\usepackage{adjustbox}
\definecolor{accent}{RGB}{0,100,148} 
\usetikzlibrary{shapes.arrows}
\usetikzlibrary{positioning}
\usepackage{calc}
\usepackage{float}
\tikzset{%
  symbol/.style={
    draw=none,
    every to/.append style={
      edge node={node [sloped, allow upside down, auto=false]{$#1$}}
    },
  },
}

\usetikzlibrary{decorations.pathmorphing}

\let\LATEXth\th

\let\th\LATEXth

\usepackage{caption}
\newtheorem{theorem}{Theorem}[section]
\newtheorem{lemma}[theorem]{Lemma}

\newtheorem{proposition}[theorem]{Proposition}
\newtheorem{corollary}[theorem]{Corollary}
\newtheorem*{summary*}{Summary}
\newtheorem{conjecture}{Conjecture}

\theoremstyle{definition}
\newtheorem{definition}[theorem]{Definition}
\newtheorem{example}[theorem]{Example}

\theoremstyle{remark}
\newtheorem{remark}[theorem]{Remark}

\numberwithin{equation}{section}

\DeclareMathOperator{\Hol}{Hol} 
 
\DeclareMathOperator{\End}{End} 

\usepackage{amsmath,amssymb}

\newcommand{\Z}{\mathbb{Z}}

\newcommand{\Herm}{\mathrm{Herm}}
\newcommand{\id}{\mathrm{id}}
\newcommand{\Sep}{\mathrm{Sep}}

\newcommand{\Ad}{\mathrm{Ad}}
\newcommand{\Ind}{\mathrm{Ind}}
\newcommand{\ind}{\mathrm{ind}}
\newcommand{\sgn}{\mathrm{sgn}}
\newcommand{\ad}{\mathrm{ad}}

\newcommand{\sep}{\mathsf{sep}}  

\usepackage[utf8]{inputenc}
\usepackage{bxcjkjatype}
\usepackage{multirow}

\definecolor{QGLBlue}{RGB}{65,105,225}   
\definecolor{QGLRed}{RGB}{203,65,84}     
\definecolor{QGLGreen}{RGB}{34,139,34}   

\newcommand{\hlEq}[1]{%
  \begingroup\setlength{\fboxsep}{2pt}%
  \colorbox{QGLBlue!5}{$\displaystyle #1$}%
  \endgroup}

\definecolor{QGLItem}{RGB}{255,245,204} 

\usepackage[most]{tcolorbox}
\definecolor{QGLItem}{RGB}{65,105,225}
\newtcolorbox{HLblock}{enhanced,breakable,
  colback=QGLBlue!5, colframe=QGLBlue!5, 
  boxrule=0pt, arc=0pt, outer arc=0pt,
  left=2pt,right=2pt,top=1.5pt,bottom=1.4pt}

\begin{document}

\title{Quantum Entanglement as a Cohomological Obstruction}

\author{Kazuki Ikeda}
\address{}
\curraddr{}

\email{kazuki.ikeda@umb.edu}
\address{Department of Physics, University of Massachusetts Boston, USA}
\address{Center For Nuclear Theory, Department of Physics and Astronomy, Stony Brook University, USA}

\thanks{The author expresses gratitude to Steven Rayan for his careful reading of the manuscript and for his invaluable comments. The author is also grateful to Myungbo Shim for useful discussions. This work was partially supported by the NSF under Grant No. OSI-2328774.}

\subjclass[2020]{81P40 (Primary),
14D24, 
14F05, 
58J20, 
53C05 
}
\date{}

\dedicatory{}
\begin{abstract}
We recast quantum entanglement as a cohomological obstruction to reconstructing a global quantum state from locally compatible information. We address this by considering presheaf cohomologies of states and entanglement witnesses. Sheafification erases the global-from-local signature while leaving within-patch multipartite structure, captured by local entanglement groups introduced here. For smooth parameter families, the obstruction admits a differential-geometric representative obtained by pairing an appropriate witness field with the curvature of a natural unitary connection on the associated bundle of amplitudes. We also introduce a \emph{Quantum Entanglement Index} (QEI) as an index-theoretic invariant of entangled states and explain its behavior. Finally, we outline a theoretical physics approach to probe these ideas in quantum many-body systems and suggest a possible entanglement-induced correction as an experimental validation. Detailed numerical implementations for concrete quantum many-body models are presented in the companion paper~\cite{2026arXiv260113467I}.
\end{abstract}
\maketitle
\tableofcontents
\section{Introduction}
\noindent
\textbf{Preliminaries.}
Quantum entanglement is the organizing principle of quantum systems. It encodes correlations that cannot be reproduced by classical probabilistic mixtures. In an entangled state, information is shared across the whole system in a way that is not reducible to correlations of its subsystems. This indicates that local information obtained by tracing out subsystems need not determine the global state. Motivated by this thought, we test entanglement through a reconstruction question: when do locally compatible marginals assemble into a global state, and when is such a state unique? In algebro-geometric terms, the presheaf of states may fail to satisfy the sheaf axioms in the presence of entanglement: gluing can fail, and even when gluing is possible, uniqueness can fail. In this setting, entanglement appears as a cohomological obstruction to global reconstruction. Reduced density operators may agree on all overlaps while a global state is not unique.

A variety of quantitative measures of entanglement have been introduced, ranging from entanglement entropy for pure states, to entanglement of formation, distillable entanglement, and relative entropy of entanglement for mixed states, to more computationally motivated quantities such as concurrence and negativity (see, e.g., \cite{Horodecki:1996nc,2001quant.ph..1032T,Guhne:2008qic,umegaki1962conditional}). These are real‑valued functions of states and are useful for analysis and numerics, but they are not directly accessible experimentally. By contrast, we work with entanglement witnesses, which are Hermitian observables. As physical observables, their expectation values provide testable certificates for separability or inseparability. We also use their global organization to examine compatibility across patches.

The aim of this work is to develop an algebro-geometric and topological framework in which entanglement is expressed as a cohomological obstruction to gluing local data. We proceed in two parts. On the discrete side, we organize states and entanglement witnesses into presheaves and express obstructions as Čech cohomology classes. On the geometric side, for smooth parameter families of states we obtain a Chern-Weil representative by pairing a parallel family of witnesses with the curvature of a natural unitary connection on the associated bundle of amplitudes (see \cite{UHLMANN1986229,Petz:1999xrh,BottTu}). This links the reconstruction problem to standard objects in differential geometry. The result is a correspondence between compatibility data and characteristic classes on parameter space, developed in a finite‑dimensional setting and suggesting broader interactions with geometry.

\vskip0.3cm
\noindent
\textbf{The gluing problem.}
Fix a finite index set $I=\{1,\dots,N\}$ and finite dimensional complex Hilbert spaces $H_j\simeq\mathbb C^{d_j}$. For every nonempty $U\subseteq I$, put
\[
H_U:=\bigotimes_{j\in U}H_j,\qquad
\mathcal D(U):=\{\rho\in\mathrm{End}(H_U)\mid \rho=\rho^\dagger,\ \rho\ge0,\ \mathrm{Tr}\,\rho=1\}.
\]
Let $\mathcal S\subset \mathcal D$ be any restriction-stable presheaf of quantum states: for each $U\subseteq I$, $\mathcal S(U)\subseteq\mathcal D(U)$ and for all $V\subseteq U$ one has $\mathrm{Tr}_{U\setminus V}\,\mathcal S(U)\subseteq\mathcal S(V)$ (partial trace on the discarded factor is a common choice for a restriction map). With restriction maps $r^U_{V}:S(U)\to S(V)$ for $V\subset U$, define the $0$-cochains by $C^0(\mathcal U,S):=\prod_{i} S(U_i)$ and the $1$-cochains $C^1(\mathcal U,S):=\prod_{i<j} S(U_{ij})$ with $U_{ij}:=U_i\cap U_j$. Let $H^0$ be the equalizer of the two families of restriction maps:
\[
H^0(\mathcal U,S)
:=\Big\{(s_i)_i\in\textstyle\prod_i S(U_i)\ :\ r^{U_i}_{U_{ij}}(s_i)=r^{U_j}_{U_{ij}}(s_j)\ \text{ for all }i<j\Big\}.
\]

When each $S(U)$ is an abelian group and the restrictions are homomorphisms, $\delta^0:C^0(\mathcal U,S)\to C^1(\mathcal U,S)$ is
\[
(\delta^0 s)_{ij}\;:=\; r^{U_i}_{U_{ij}}(s_i)\;-\; r^{U_j}_{U_{ij}}(s_j)\qquad (i<j),
\]
and
\[
H^0(\mathcal U,S)=\ker\delta^0.
\]

\begin{example}
Let $\mathcal V(U)=\mathrm{Herm}(H_U)$ equipped with the trace pairing $\langle X,Y\rangle = \mathrm{Tr}(XY)$. We consider the usual restriction given by partial trace. Writing $X=(X_i)_i\in\prod_i \mathcal V(U_i)$, for each overlap $U_{ij}$ we have
\[
(\delta^0 X)_{ij}
\;=\;
\Tr_{\,U_i\setminus U_{ij}}(X_i)\;-\;\Tr_{\,U_j\setminus U_{ij}}(X_j)
\;\in\; V(U_{ij}).
\]
Thus
\[
H^0(\mathcal U,\mathcal V)
=\Big\{(X_i)_i:\ \Tr_{\,U_i\setminus U_{ij}}(X_i)=\Tr_{\,U_j\setminus U_{ij}}(X_j)\ \text{for all }i<j\Big\}.
\]
\end{example}

\vskip0.3cm
The map
\[
j_{\mathcal U}:\ \mathcal S(I)\longrightarrow H^0(\mathcal U,\mathcal S),\qquad
\rho\longmapsto \bigl(\rho|_{U_i}\bigr)_i
\]
encodes the sheaf axiom on $(\mathcal U,\mathcal S)$. Its failure appears in two distinct modes:
\begin{HLblock}
\begin{enumerate}[label=\textbf{(\Alph*)},leftmargin=*]
\item\label{A}\textbf{Gluing failure:} $j_{\mathcal U}$ is not surjective. There exist compatible marginals $\{\rho_i\in\mathcal S(U_i)\}_i$ (i.e. agreeing on all overlaps) which admit no global $\rho\in\mathcal S(I)$ with $\mathrm{Tr}_{I\setminus U_i}\rho=\rho_i$ for every $i$.
\item\label{B}\textbf{Non-uniqueness:} $j_{\mathcal U}$ is not injective. There exist two distinct global states
$\rho_1\neq\rho_2\in\mathcal S(I)$ with identical marginals on every $U_i$.
\end{enumerate}
\end{HLblock}
Whether \ref{A} or \ref{B} occurs depends on the choice of $\mathcal S$. When $\mathcal S$ is a sheaf, both pathologies are absent by definition. In physically relevant presheaves, however, entanglement forces one (or both) to appear. For instance, for $\mathcal S=\mathcal D$ and $I=\{1,2\}$ the four Bell states $\ket{\Psi^\pm}=\frac{\ket{01}\pm\ket{10}}{\sqrt{2}},\ket{\Phi^\pm}=\frac{\ket{00}\pm\ket{11}}{\sqrt{2}}$ share the same single–site marginals (\(\Tr_{\{i\}}\ket{\Psi^\pm}\bra{\Psi^\pm}=\Tr_{\{i\}}\ket{\Phi^\pm}\bra{\Phi^\pm}=\frac{\mathbf 1}{2},~i\in\ I\)), so the restriction $\mathcal D(I)\to\mathcal D(\{1\})\times\mathcal D(\{2\})$ is not injective. This is an instance of \ref{B}. For the broader theme of uniqueness of state extensions, see \cite{25f8c26d-c81c-39c1-9128-a71146ac84ff,33675660-3119-37a8-a5d2-fa5a40dfb227,7f2d66a9-2d0a-3ab6-8756-1cda100696a6}.

An instructive example occurs for the presheaf $P\subset D$ of pure states. If $\mathcal U=\{U_i\}_{i=1}^m$ is a partition of $I$ (i.e.\ the $U_i$ are disjoint and $\bigsqcup_i U_i=I$), then any family of pure states $\rho_i\in P(U_i)$ glues uniquely to the tensor product state $\rho=\bigotimes_{i=1}^m \rho_i$. No overlap-compatibility conditions arise because $U_i\cap U_j=\varnothing$. By contrast, for a genuine cover with nontrivial overlaps, requiring a \emph{global pure} extension is restrictive: if a reduced state on an overlap $U_{ij}$ is pure, then the corresponding local vector must factor across the cut $U_{ij}\mid (U_i\setminus U_{ij})$, and a further $U(1)$-valued Čech cocycle records the residual phase obstruction. See Lemma~\ref{lem:pure-ordered} for the precise statement.

\begin{proposition}\label{prop:pure-fail}
Let $\mathcal U=\{U_i\}_{i=1}^m$ be a partition of $I$. Given pure local states $\rho_i=|\psi_i\rangle\!\langle\psi_i|\in P(U_i)$, there exists a unique global pure state $\rho=|\Psi\rangle\!\langle\Psi|\in P(I)$ with $\rho|_{U_i}=\rho_i$ for all $i$, namely $|\Psi\rangle=\bigotimes_{i=1}^m |\psi_i\rangle$. Conversely, for a global pure state $\rho=|\Psi\rangle\!\langle\Psi|$ the following are equivalent: (i) $\rho|_{U_i}$ is pure for every $i$; (ii) $|\Psi\rangle=\bigotimes_{i=1}^m |\psi_i\rangle$ (a product across the partition).
\end{proposition}

Our approach is entirely in the density-matrix (operator) formalism: we pose gluing as a problem for a presheaf of state spaces and measure its failures by cohomological obstructions \(R^0\) (Definition \ref{def:compat-defect}) for non–uniqueness \ref{B} and the local entanglement groups \(E^q\) (Definition \ref{def:Eq-ordered-cosimplicial}) for within-patch multipartite content \ref{A} with operational (separable-witness) certificates. This complements the spectral/representation theoretic marginal program, where necessary constraints on compatible spectra are obtained via moment map \cite{guillemin1982convexity,kirwan1984convexity,berenstein2000coadjoint} and via marginal spectra analyses \cite{Christandl2006,2014CMaPh.332....1C}. For a math treatment linking quantum marginals to projections of coadjoint/orbital measures, see Collins–McSwiggen \cite{Collins:2021xyq}. Our approach is orthogonal to these methods. Conceptually, our use of presheaves echoes the sheaf-theoretic account of contextuality/nonlocality as obstructions to global sections \cite{abramsky2011sheaf}, though here the presheaf consists of density operators. More specifically, our obstruction $R^0$ measures the degree of non-uniqueness, while the invariants $E^q$ capture the local entanglement that either prevents or allows gluing. Through the witness characterization, these quantities are equivalent to determining whether certain pairings with separable witnesses vanish or not. This is discussed from the perspective of quantum information theory. The use of entanglement witnesses as separating functionals for the cone of separable states goes back to the Horodecki criterion \cite{Horodecki:1996nc}, with systematic expositions by Terhal \cite{2001quant.ph..1032T} and G{\"u}hne–T{\'o}th \cite{Guhne:2008qic}.

\vskip0.3cm
\noindent
\textbf{Witness presheaf and obstructions.}
We encode entanglement by witnesses. For $U\subseteq I$ let the witness cone be
\begin{equation}\label{eq:Witt}
   \mathsf{Wit}(U)
   :=\bigl\{W=W^\dagger\in\End(H_U)\ \big|\ 
       \Tr(W\sigma)\ge0\ \ \forall\,\sigma\in\Sep(U)\bigr\},
\end{equation}
and put $\mathcal W(U):=\mathrm{span}_{\mathbb R}\mathrm{Wit}(U)$. Here \(\Sep(U)\) is the set of all separable states on $U$. For an inclusion \(V\subseteq U\) define the restriction
\begin{equation}\label{eq:w-restr}
   r^U_{V}:\mathcal W(U)\to\mathcal W(V),\qquad
   r^U_{V}(W):=\Tr_{U\setminus V}\!\bigl[W\,(\mathbf 1_V\otimes\tau_{U\setminus V})\bigr],
\end{equation}
with a fully separable state \(\tau\) on the traced factor. We fix, once and for all, a sitewise faithful state $\{\tau_j\}_{j\in I}$ and set $\tau_S := \bigotimes_{j\in S}\tau_j$ for every $S\subset I$. With this multiplicative rule, we see that the witness restriction
$r^{U}_{V}(W):=\Tr_{U\setminus V}\!\big[\,W\,(\mathbf{1}_V\otimes \tau_{U\setminus V})\big]$
satisfies the axiom: for $U_3\subset U_2\subset U_1$,
\[
r^{U_2}_{U_3}\circ r^{U_1}_{U_2} \;=\; r^{U_1}_{U_3}.
\]
Indeed, functoriality follows from $\tau_{U_1\setminus U_3}=\tau_{U_2\setminus U_3}\otimes \tau_{U_1\setminus U_2}$ and the cyclicity of the trace. This fixed choice will be used throughout.

We use $W$ as functionals on $\mathcal V(U):=\mathrm{Herm}(H_U)$ via the trace pairing. In finite dimension, one has $\mathcal{W}(U)=\mathcal V(U)$ (Proposition~\ref{prop:wit-separate-span}), and the linear Čech complex contracts in positive degrees by product state extensions (Proposition~\ref{prop:degree0}). Consequently, $H^{>0}(\mathcal U, \mathcal{W})=0$ for every finite cover $\mathcal U$, and there is no nontrivial group-valued ``existence obstruction'' at degree~1. Existence of a global state is instead a cone-feasibility question certified by duality: the certificates for positive semi-definite (PSD) and separable infeasibility are given in Propositions~\ref{prop:state-level-extension} and~\ref{prop:farkas}, which are extensions of Farkas's lemma. Uniqueness is measured by the kernel $R^0=\ker j$ of the degree–0 map. Operationally, vanishing of a Čech class built from local states is equivalent to pairing to zero against every compatible family in \(C^\bullet(\mathcal U,\mathcal W)\) (Theorem~\ref{thm:witness-vanish}).

For each region $U$ we view $\mathcal S(U)\subset\mathcal D(U)$ as a presheaf of states with restriction by partial trace $r^U_V:\mathcal S(U)\to\mathcal S(V)$ for $V\subset U$. The corresponding witness cone (Definition \ref{def:dualcone})
\[
\mathrm{Wit}(U):=C^*_{\mathrm{sep}}(U)
\]
and its real linear span is $W(U):=\operatorname{span}_{\mathbb{R}}\mathrm{Wit}(U)$.
Witnesses pull back along inclusions by the adjoints $r_U^V{}^{\dagger}$ so that the trace pairing is natural: $\langle r_U^V{}^{\dagger}W,\,X\rangle=\langle W,\,r_U^V X\rangle$.

Within a single patch, the ancilla column built from state insertion and single-slot resets has cohomology $E^q(U)$. However, the signed stabilize--reset pairing does not detect nontriviality of a cohomology class: for any $Y$ one has
\[
\sum_{m=0}^{q+1}(-1)^m \Tr\bigl(W\, d_q^{(m)}(Y)\bigr)=\Tr\bigl(W\,\delta_E(Y)\bigr),
\]
so it vanishes whenever $Y$ is $\delta_E$--closed (Definition~\ref{def:witness-trivial}). Operational entanglement detection is therefore captured
by LED$(q)$, which tests whether the obstruction $\delta_E(Y)$ can be made negative against a separability witness. For the reindexed differential, LED$(q)$ exhibits a parity collapse: it is impossible for even $q\ge 1$, while for odd $q\ge 1$ it is equivalent to an ordinary $q=0$ witness test (Proposition~\ref{prop:LED-tau-independence}).

\vskip0.3cm
\noindent
\textbf{Sheafification.}
Sheafification of a presheaf $\mathcal S$ produces a sheaf $\mathcal S^\#$ and a natural map $\eta:\mathcal S\to\mathcal S^\#$ with the property that, on any fixed cover $\mathcal U=\{U_i\}$, the degree-$0$ restriction $j^\#_{\mathcal U}:\mathcal S^\#(U)\!\xrightarrow{\;\sim\;} H^0(\mathcal U,\mathcal S^\#)$ is bijective. Thus all gluing and separatedness issues visible on $\mathcal U$ are removed, and any two global states that are indistinguishable by all restrictions to $\mathcal U$ (and its refinements) become equal in $\mathcal S^\#$. However, this does not mean that the underlying state is separable: sheafification erases only the ``global vs.\ local'' signature of entanglement. Patchwise multipartite structure (as measured, e.g., by the local entanglement groups $E^q$ (Definition \ref{def:Eq-ordered-cosimplicial}) attached to a single patch) may still be present, but it is invisible to the local tests. Therefore $\overline{E}^q=\varinjlim_{\mathcal U} E^q(\mathcal U)$ can remain non‑zero after sheafification (Theorem \ref{thm:completeness-weak}). In particular, every computation that depends solely on local projections and their linear combinations (e.g.\ $\sum_i \Tr[W_i\,\rho|_{U_i}]$) factors through $\eta$ and is therefore preserved under sheafification on that cover.

A presheaf on a finite cover behaves as a distributed data structure: each $U_i$ stores a local record. Sheafification performs a natural consistency completion and a quotient by observational equivalence. Here any two records that produce the same observable outcomes under projections are considered equivalent. After sheafification, the remaining invariants live in the classical side. By contrast, non-classical features we study here reside in the presheaf picture prior to sheafification.

\vskip0.3cm
\noindent
\textbf{Quantum Entanglement Index.}
Section~\ref{sec:obstruction} organizes states and witnesses into a presheaf and expresses ``global-from-local'' failures as Čech classes. When the local data vary smoothly over a parameter manifold $X$, the obstruction acquires a differential geometric description. Section~\ref{sec:diff-form-obstruction} then explores these obstructions from the perspective of differential geometry on parameter space: given a smooth family of full‑rank states \(\rho:X\to\mathcal D_{\mathrm{full}}(H)\), we pass to the pullback amplitude bundle
\[
E_\rho:=\{(x,W)\in X\times GL(H): WW^\dagger=\rho(x)\}\to X,
\]
which is a principal \(U(r)\)–bundle.  In this setting the Čech cocycle built from local amplitudes and their unitary transitions is compared with de~Rham cohomology: for a smooth $A$–parallel witness field $W$, the ordered pairing equals the Chern–Weil form (Theorem \ref{thm:global-uhl}):
\[
\big\langle W,\ [\delta^k(u)]_{\check C}\big\rangle
\;=\;
\Big[\,(2\pi i)^{-k}\,\Tr\big(WF_A^{\,k}\big)\,\Big]_{\!dR},
\]
so the obstruction class is represented on $X$ by the closed forms $(2\pi i)^{-k}\Tr(WF_A^{k})$. Physically, it can be understood as an extension of Berry/TKNN number (see Section \ref{sec:phys}).

This geometric counterpart converts the entanglement obstructions into characteristic classes on $X$ and furnishes the input for the index‑theoretic refinement that follows in Section~\ref{sec:QGL}. Let $(E,A)$ be a Hermitian bundle with unitary connection and let $W\in\Gamma(\mathrm{End}(E))$ be an $A$–parallel witness field on $X$. Writing $S:=\mathrm{sgn}(W)$ and $F_A$ for the curvature, the \emph{Quantum Entanglement Index} (QEI) (Definition~\ref{def:quantum-index}) records the difference of Dirac indices seen by the $S$–positive/negative sectors:
\[
\mathrm{Ind}_S\big(D_X\!\otimes E\big)
:= \mathrm{ind}\big(D_X\!\otimes E_+\big) - \mathrm{ind}\big(D_X\!\otimes E_-\big)
= \Big\langle \widehat A(TX)\wedge \mathrm{Tr}\big(Se^{F_A/2\pi i}\big),[X]\Big\rangle .
\]
Thus an $A$–parallel witness grades the theory by an integer, and this grading is topological on parameter space:
\begin{figure}[H]
\centering
\begin{tikzcd}[column sep=huge]
\mathsf{QGL}(X):=\{(E,A,W):D_AW=0\}
\arrow[r, dashed]
& \mathcal{D}\mathrm{-mod}(\mathrm{Bun}_G)
\end{tikzcd}
\end{figure}

\noindent
We use this index to grade automorphic data: to a spectral datum $(E,A,W)$ we associate an automorphic object together with its $\mathbb Z$–degree $\mathrm{Ind}_S$. In this way, the quantum entanglement refinement overlays the classical automorphic category with a natural grading determined by the witness. A Hecke functor then has a quantized effect on the grading. If $H_{x,\lambda}$ is a Hecke modification of coweight $\lambda$ at $x$, its action shifts the QEI by the signed charge $\langle S,\lambda\rangle$:
\[
\mathrm{Ind}_S\!\left(H_{x,\lambda}\,\mathsf{Aut}_{(E,W)}\right)
=\mathrm{Ind}_S\!\left(\mathsf{Aut}_{(E,W)}\right)+\langle S,\lambda\rangle .
\]
The correspondence can be summarized in the following diagram:
\begin{figure}[H]
\centering
\begin{tikzcd}[column sep=large]
& \mathrm{Hecke}_x \arrow[dl,"p_1"'] \arrow[dr,"p_2"] & \\
\mathsf{Aut}_{(E,W)} \arrow[rr, dashed, "H_{x,\lambda}"] \arrow[d, "\mathrm{Ind}_S" left] 
&& \mathsf{Aut}_{(E,W)} \arrow[d, "\mathrm{Ind}_S"] \\
\mathbb Z \arrow[rr, "{+\langle S,\lambda\rangle}"] && \mathbb Z
\end{tikzcd}
\end{figure}
\noindent
This expresses that Hecke modifications induce discrete jumps between entanglement sectors, which can be interpreted as an entanglement-induced quantum phase transition in quantum many-body systems. The automorphic object is transformed while its degree is shifted by $\langle S,\lambda\rangle$. The index also admits a differential–geometric realization.

\vskip0.3cm
\noindent
\textbf{Bridging to geometric Langlands.}
We begin with fundamental structures of the geometric Langlands correspondence. It is expected that flat ${}^{L}\!G$–local systems on $X$ correspond to $\mathcal{D}$–modules on $\mathrm{Bun}_G$:
\begin{figure}[H]
\centering
\begin{tikzcd}[column sep=large]
\mathrm{Loc}_{\,^{L}\!G}(X) \arrow[r, dashed, "\mathrm{GL}"] 
& \mathcal{D}\mathrm{-mod}(\mathrm{Bun}_G)
\end{tikzcd}
\end{figure}
\noindent
Then the Hecke correspondence at a point $x\in X$ gives endofunctors of the automorphic category by convolution with spherical kernels. 
Concretely, for every $V\in\mathrm{Rep}({}^{L}\!G)$ one has the functor
\[
H_{x,V}\;=\;p_{2!}\!\left(p_1^{*}(-)\otimes\mathcal S_{V}\right),
\]
assembled from the correspondence:
\begin{figure}[H]
\centering
\begin{tikzcd}
& \mathrm{Hecke}_x \arrow[dl,"p_1"'] \arrow[dr,"p_2"] & \\
\mathrm{Bun}_G \arrow[rr, dashed,
"{H_{x,V}=\,p_{2!}\big(p_1^{*}(-)\otimes\mathcal S_{V}\big)}"'] 
&& \mathrm{Bun}_G
\end{tikzcd}
\end{figure}

\noindent
For a local system $E$, the associated automorphic object $\mathrm{Aut}_E$ is a Hecke eigensheaf: for every $x\in X$ and $V\in\mathrm{Rep}({}^{L}\!G)$,
\[
H_{x,V}(\mathrm{Aut}_E)\;\simeq\;(V_{E_x})\otimes \mathrm{Aut}_E,
\]
so the fiber $V_{E_x}$ plays the role of the eigenvalue on the automorphic side.

\bigskip
\noindent
Finally, the spherical geometric Satake equivalence identifies the geometric inputs for the kernels: $G(\mathcal O)$–equivariant (perverse) sheaves (or $\mathcal{D}$–modules) on the affine Grassmannian are equivalent to $\mathrm{Rep}({}^{L}\!G)$. Under this identification, the kernel $\mathcal S_V$ corresponds to the representation $V$.
\begin{figure}[H]
\centering
\begin{tikzcd}[column sep=huge]
\mathsf{Perv}_{G(\mathcal O)}(\mathrm{Gr}_G) 
\arrow[r, "\mathrm{Satake}", "\simeq"'] 
& \mathrm{Rep}\!\left(^{L}\!G\right)
\end{tikzcd}
\end{figure}

\noindent
In our quantum refinement (Section~\ref{sec:QGL}), the same Hecke functors would act on automorphic data equipped with the integral grading given by the QEI, and Hecke modifications of coweight~$\lambda$ shift the degree by $\langle S,\lambda\rangle$. While the traditional framework remains intact, our approach here introduces a natural integer grading on the automorphic side that is sensitive to the entanglement response. 

The concept is further investigated in the context of topological matter \cite{2026arXiv260113467I}, where the correspondence is demonstrated numerically.

\section{\label{sec:obstruction}Quantum entanglement from the viewpoint of topology}
\subsection{\label{sec:presheaf}Presheaf of quantum states}
For each finite subsystem $U \subseteq I$, let
\[
   \mathcal V(U) := \mathrm{Herm}(H_U)
\]
be the real vector space of Hermitian operators on $H_U$, equipped with the trace pairing $\langle X,Y\rangle = \mathrm{Tr}(XY)$.

\begin{definition}\label{def:dualcone}
For each finite subsystem $U\subseteq I$, define the separable cone
\[
  C_{\mathrm{sep}}(U)\ :=\ \mathrm{cone}\{\rho_{\mathrm{sep}}\in\mathcal D_{\mathrm{sep}}(U)\},
\]
and its dual cone
\[
  C_{\mathrm{sep}}^\ast(U)\ :=\ \{\,W\in \mathcal V(U):\ \Tr(W\sigma)\ge 0\ \ \forall\,\sigma\in C_{\mathrm{sep}}(U)\,\}.
\]
Elements of $C_{\mathrm{sep}}^\ast(U)$ are entanglement witnesses: $W$ detects entanglement in
$\rho$ if $\Tr(W\rho)<0$.
\end{definition}

\noindent
$C_{\mathrm{sep}}(U)$ is closed, convex, and pointed, and partial trace is positive and sends $C_{\mathrm{sep}}(U)$ into $C_{\mathrm{sep}}(V)$ for $V\subseteq U$. Thus $(V,C_{\mathrm{sep}})$ is a presheaf of ordered vector spaces.

\begin{proposition}\label{prop:wit-separate-span}
For every finite $U\subseteq I$:
\begin{enumerate}\itemsep.25em
\item \(\mathrm{Int}\,C_{\mathrm{sep}}^\ast(U)\neq\varnothing\) and \(\mathbf 1\in \mathrm{Int}\,C_{\mathrm{sep}}^\ast(U)\).
\item \(\displaystyle \bigcap_{W\in C_{\mathrm{sep}}^\ast(U)} \ker\langle W,\cdot\rangle \;=\; \{0\}\).
Equivalently,
\[
  J(U)\ :=\ \{X\in \mathcal V(U):\ \Tr(WX)=0\ \forall\,W\in C_{\mathrm{sep}}^\ast(U)\}\ =\ \{0\}.
\]
\item \(\displaystyle \operatorname{span}_{\mathbb R} C_{\mathrm{sep}}^\ast(U)\;=\;\mathcal V(U)\).
In particular, \(W(U):=\operatorname{span}_{\mathbb R}\mathrm{Wit}(U)=\mathcal V(U)\).
\end{enumerate}
\end{proposition}

\begin{proof}
(1) For every nonzero \(\sigma\in C_{\mathrm{sep}}(U)\) we have \(\Tr(\sigma)>0\). Hence \(\Tr(\mathbf 1\cdot\sigma)>0\) for all \(\sigma\in C_{\mathrm{sep}}(U)\setminus\{0\}\), which by the standard characterization of interiors of dual cones (in finite dimension) places \(\mathbf 1\) in \(\mathrm{Int}\,C_{\mathrm{sep}}^\ast(U)\).

(2) Let \(X\neq 0\). If \(\Tr(\mathbf 1\,X)\neq 0\) then \(\mathbf 1\) already separates \(X\). Otherwise, choose \(|t|\) small and set \(W:=\mathbf 1+tX\). Since \(\mathbf 1\in \mathrm{Int}\,C_{\mathrm{sep}}^\ast(U)\), such \(W\) still lies in \(C_{\mathrm{sep}}^\ast(U)\), while \(\Tr(WX)=t\,\Tr(X^2)\neq 0\) (trace pairing on Hermitians). Thus the intersection of kernels is \(\{0\}\).

(3) A convex cone with nonempty interior is full–dimensional, hence the linear span of \(C_{\mathrm{sep}}^\ast(U)\) equals \(\mathcal V(U)\). Identifying \(\mathcal V(U)\cong \mathcal V(U)^\ast\) by the trace pairing gives the claim for \(W(U)\) as well.
\end{proof}

Consequently, the separability test is as follows:
\[
   \rho\in\mathcal D(U)\ \text{is separable}
   \quad\Longleftrightarrow\quad
   \Tr(W\rho)\ge 0 \ \ \text{for all } W\in C_{\sep}^\ast(U).
\]
Separable states are exactly those that are nonnegative on all witnesses. Accordingly, vanishing of obstruction classes is characterized by pairing to zero with all compatible families in the linear span $W(\bullet)$ (Theorem~\ref{thm:witness-vanish}).

\subsection{Entanglement obstructions}
\begin{definition}
\label{def:compat-defect}
Let $\mathcal U=\{U_i\}_{i\in\mathcal I}$ be a finite cover. For a family of local states $\sigma=(\sigma_i)_{i\in\mathcal I}\in C^0(\mathcal U,\mathcal D)$ define its (linear) compatibility defect
\[
\Delta(\sigma):=\delta\sigma\ \in\ C^1(\mathcal U,\mathcal V).
\]
Let $j:\mathcal V(I)\to H^0(\mathcal U, \mathcal V)$ be the degree-0 restriction map $X\mapsto (X|_{U_i})_i$. We define
\[
  Q^{0}(\mathcal U):=\frac{H^0\!\big(\mathcal U,\mathcal V\big)}{\,j\big(\mathcal V(I)\big)}\!,
  \qquad
  R^{0}(\mathcal U):=\ker\big(j:\mathcal V(I)\to H^0(\mathcal U,\mathcal V)\big).
\]
\end{definition}

\begin{proposition}
\label{prop:degree0}
There exists a Čech contracting homotopy built from product–state extensions $e_{VU}:X\mapsto X\otimes\tau_{U\setminus V}$. Hence $H^{k>0}(\mathcal U,\mathcal V)=0$ and $j$ is surjective, so $Q^0(\mathcal U)=0$.

Moreover, for a local-state family $\sigma$ one has
\[
\Delta(\sigma)=0\ \Longleftrightarrow\ \exists\,X\in \mathcal V(I)\ \text{with}\ X|_{U_i}=\sigma_i\quad\forall i.
\]
\end{proposition}
\begin{proof}
Fix a finite cover $\mathcal U=\{U_i\}_{i=1}^m$ of $I$. For every inclusion $V\subset U$ define the linear extension
\[
e_V^U:V(V)\to V(U),\qquad e_V^U(X):=X\otimes \tau_{U\setminus V},
\]
where $\tau_{U\setminus V}:=\bigotimes_{j\in U\setminus V}\tau_j$ and $\Tr(\tau_{U\setminus V})=1$. Then the restriction $r_V^U=\Tr_{U\setminus V}$ satisfies
\[
r_V^U\circ e_V^U=\id_{V(V)}.
\]
Hence each restriction map is split surjective (the presheaf is flabby on this finite site).

Let $\check C^p(\mathcal U,V)=\prod_{i_0<\cdots<i_p}V(U_{i_0\cdots i_p})$ with the usual Čech differential $\delta$. Choose the distinguished index $1$ and define, for $p\ge 1$, a degree $(-1)$ map $h^p:\check C^p(\mathcal U,V)\to \check C^{p-1}(\mathcal U,V)$ by
\[
(h^p c)_{i_0\cdots i_{p-1}} :=
\begin{cases}
e^{U_{i_0\cdots i_{p-1}}}_{U_{1 i_0\cdots i_{p-1}}}\!\bigl(c_{1 i_0\cdots i_{p-1}}\bigr), & 1<i_0,\\[2mm]
0, & i_0=1,
\end{cases}
\]
where $U_{i_0\cdots i_{p-1}}:=U_{i_0}\cap\cdots\cap U_{i_{p-1}}$ and
$U_{1 i_0\cdots i_{p-1}}:=U_1\cap U_{i_0}\cap\cdots\cap U_{i_{p-1}}$. A standard computation using only $r\circ e=\id$ and the functoriality of the extensions shows that
\[
\delta h^p + h^{p+1}\delta=\id
\qquad\text{for all }p\ge 1.
\]
Therefore the Čech complex is contractible in positive degrees, so $H^{k>0}(\mathcal U,V)=0$.

In particular, the degree--$0$ restriction map $j:V(I)\to H^0(\mathcal U,V)$ is surjective, hence $Q_0(\mathcal U)=0$. Finally, for $\sigma\in \check C^0(\mathcal U,V)$ one has $\delta\sigma=0$ if and only if $\sigma$ is a $0$--cocycle, and then the contracting homotopy yields an $X\in V(I)$ with $X|_{U_i}=\sigma_i$ for all $i$.
\end{proof}

By non‑degeneracy of the trace pairing on $C^1(\mathcal U,\mathcal V)$, 
\[
\delta\sigma=0\quad\Longleftrightarrow\quad \langle W,\delta\sigma\rangle=0\ \ \text{for all }\,W\in C^1(\mathcal U,\,W(\bullet)).
\]

\subsection{Feasibility by duality}
\begin{proposition}
\label{prop:state-level-extension}
Let $\sigma=(\sigma_i)\in C^0(\mathcal U,\mathcal{D})$ be compatible ($\Delta(\sigma)=0$).

\emph{(PSD extension).} 
If there exist Hermitian operators $Y_i$ and $\alpha\in\mathbb R$ such that
\[
\Xi:=\sum_i r_i^{\!*}(Y_i)+\alpha\,\mathbf 1 \succeq 0, \qquad \sum_i \operatorname{Tr}(Y_i\,\sigma_i)+\alpha<0,
\]
then no global density matrix $\rho\in\mathcal{D}(I)$ realizes the marginals $\sigma_i$\footnote{For Hermitian operators $A,B$ we write $A\succeq B$ if $A-B$ is positive semidefinite, and $A\succ B$ if $A-B$ is positive definite. In particular $A\succeq 0$ means $A$ is PSD.}.

\emph{(Separable extension).} 
If there exist witnesses $W_i\in C^*_{\mathrm{sep}}(U_i)$ with
\[
\sum_i \operatorname{Tr}(W_i\,\sigma_i)<0,
\]
then no separable global state $\rho_{\mathrm{sep}}\in\mathcal{D}_{\mathrm{sep}}(I)$ realizes $\sigma$.
\end{proposition}

\begin{proof}
For each inclusion $U_i\subset I$, the restriction $r_i=r^I_{U_i}:\mathcal V(I)\to \mathcal V(U_i)$ is the partial trace, and its adjoint $r_i^{\!*}:\mathcal V(U_i)\to \mathcal V(I)$ is the extension $Z\mapsto Z\otimes \mathbf 1$ on the missing tensor factors. These satisfy the trace adjointness relation
\[
\Tr\big(r_i^{\!*}(Y)\,\rho\big) = \Tr\big(Y\, r_i(\rho)\big).
\]
Moreover, if $W\in C^*_{\mathrm{sep}}(U)$ is a separability witness, then $W\otimes \mathbf 1\in C^*_{\mathrm{sep}}(I)$, since
\[
\Tr\big((W\otimes \mathbf 1)\,\rho_{\mathrm{sep}}\big) = \Tr\big(W\,\Tr_{I\setminus U}\rho_{\mathrm{sep}}\big)\ge 0
\]
for every separable $\rho_{\mathrm{sep}}\in \mathcal D_{\mathrm{sep}}(I)$.

\smallskip
\noindent{(PSD extension).}
Suppose by contradiction that there exists a global density matrix $\rho\in \mathcal D(I)$ with marginals $\sigma_i$. Then by adjointness,
\[
\Tr(\Xi\rho) = \sum_i \Tr(r_i^{\!*}(Y_i)\rho)+\alpha\Tr(\rho) = \sum_i \Tr(Y_i\sigma_i)+\alpha < 0.
\]
On the other hand, since $\Xi\succeq 0$ and $\rho\succeq 0$, we must have $\Tr(\Xi\rho)\ge 0$, which is a contradiction. Thus no such $\rho$ exists.

\smallskip
\noindent{(Separable extension).}
Let $W_i\in C^*_{\mathrm{sep}}(U_i)$ be as in the statement, and define $\widehat W:=\sum_i r_i^{*}(W_i)$. By the closure of $C^*_{\mathrm{sep}}(I)$ under such extensions and sums, we have $\widehat W\in C^*_{\mathrm{sep}}(I)$. If there were a separable extension $\rho_{\mathrm{sep}}\in \mathcal D_{\mathrm{sep}}(I)$ of $\sigma$, then
\[
\Tr(\widehat W\rho_{\mathrm{sep}}) = \sum_i \Tr(r_i^{*}(W_i)\rho_{\mathrm{sep}}) = \sum_i \Tr(W_i\sigma_i) < 0.
\]
But by definition of $C^*_{\mathrm{sep}}(I)$, every separability witness has nonnegative expectation on separable states. This contradiction shows that no separable extension exists.
\end{proof}

\begin{definition}
\label{def:compatible-witness-and-local}
We introduce the following terms and concepts.\\
\emph{(1) Čech cochains and the adjoint differential.}
For a finite open cover $\mathcal U=\{U_i\}_{i\in\mathcal I}$ and $k\ge 0$, set
\[
  C^{k}(\mathcal U,\mathcal V)\;:=\;\prod_{i_0<\cdots<i_k} \mathcal V\big(U_{i_0}\cap\cdots\cap U_{i_k}\big).
\]
Let $\delta:C^{k-1}(\mathcal U,\mathcal V)\to C^{k}(\mathcal U,\mathcal V)$ be the Čech coboundary built from the restriction maps $r^{(m)}$ given by partial traces over the missing factor. The adjoint $\delta^\ast:C^{k}(\mathcal U,\mathcal V)\to C^{k-1}(\mathcal U,\mathcal V)$ is defined pointwise by
\begin{equation}
\label{eq:delta_star}
  (\delta^\ast W)_{i_0\ldots i_{k-1}}\;:=\;\sum_{m=0}^{k}(-1)^m\,\big(r^{(m)}\big)^{\!*}\!\Big(W_{i_0\ldots \widehat{i_m}\ldots i_k}\Big),
\end{equation}
where $\big(r^{(m)}\big)^{\!*}$ is the extension map $Z\mapsto Z\otimes \mathbf 1$ on the traced–out tensor factor, i.e. $\langle (r^{(m)})^* Z, X\rangle=\langle Z, r^{(m)}X\rangle$ for all matching $X,Z$.

\smallskip\noindent
\emph{(2) Compatible witness families.}
A family $W\in C^{k}(\mathcal U,C_{\sep}^\ast)$ is called compatible if $\delta^\ast W=0$. Equivalently, $W$ satisfies the Čech cocycle condition with respect to the adjoint (extension) maps.

\smallskip\noindent
\emph{(3) Cochains coming from local states.}
Write $C^{k}_{+}(\mathcal U)\subset C^{k}(\mathcal U,\mathcal V)$ for the cone of state–valued $k$–cochains, i.e.\ those with each component a density operator on the corresponding overlap (positive semidefinite and trace one). We say a $k$–cochain $c\in C^{k}(\mathcal U,\mathcal V)$ comes from local states if
\begin{itemize}
  \item for $k=0$: $c\in C^{0}_{+}(\mathcal U)$ (a family of local states), and
  \item for $k\ge 1$: there exists $\sigma\in C^{k-1}_{+}(\mathcal U)$ such that
        $c=\delta\sigma$.
\end{itemize}
In particular, a cocycle $c\in Z^{k}(\mathcal U,\mathcal V)$
``coming from local states'' (with $k\ge 1$) is an obstruction cocycle of the form $c=\delta\sigma$ for some state–valued $(k{-}1)$–cochain $\sigma$.
\end{definition}

\begin{remark}
\label{rmk:witness-sheaf}
Separatedness (injectivity of $j_U$) is equivalent to: for every nonzero $[X]\in \mathcal V(U)$ there exist an index $i$ and $W_i\in W(U_i)$ such that $\mathrm{Tr}[\,W_i\, X|_{U_i}\,]\neq 0$. Surjectivity of $j_U$ (gluing) is equivalent to: if a compatible family $s=(s_i)_{i\in I}\in H^0(\mathcal U,\mathcal V)$ satisfies $\sum_i \mathrm{Tr}(W_i s_i)=0$ for every compatible family $W=(W_i)_i\in C^0(\mathcal U,\,W(\bullet))$ with $\delta^*W=0$, then $s=j_U([X])$ for some $[X]\in \mathcal V(U)$. If $j_U$ is also injective, $[X]$ is unique.
\end{remark}

\begin{theorem}
\label{thm:witness-vanish}
Let $\mathcal U=\{U_i\}_{i\in I}$ be a finite cover. Let $c\in Z^k(\mathcal U,\mathcal V)$ be a Čech $k$–cocycle. Then the following are equivalent:
\[
[c]=0 \ \Longleftrightarrow\ \langle W,c\rangle=0\quad\text{for every }W\in C^k(\mathcal U,\,W(\bullet))\text{ with }\delta^* W=0,
\]
where $W(\bullet)=\operatorname{span}_\mathbb{R}\mathsf{Wit}(\bullet)$ \eqref{eq:Witt} and $\delta^*$ is the trace–adjoint Čech coboundary \eqref{eq:delta_star}.
\end{theorem}

\begin{proof}
We first note that all vector spaces in this work are finite dimensional. For each inclusion of overlaps denote by $r^{(m)}$ the restriction (partial trace) and by $(r^{(m)})^*$ its trace–adjoint extension, so that
\begin{equation}
\label{eq:duality}
\langle W,\delta b\rangle=\langle \delta^* W,\,b\rangle    
\end{equation}
for all matching cochains $W,b$. Set $S:=\delta C^{k-1}(\mathcal U,\mathcal V)\subset C^k(\mathcal U,\mathcal V)$. By the finite dimensional duality, we have $S^\perp=\ker\delta^*$.

($\Rightarrow$). If $[c]=0$, then $c=\delta b$ for some $b$, hence $\langle W,c\rangle=\langle \delta^* W,b\rangle=0$ for all $W$ with $\delta^* W=0$ by~\eqref{eq:duality}.

($\Leftarrow$). Conversely, assume $\langle W,c\rangle=0$ for all $W\in C^k(\mathcal U,\,W(\bullet))$ with $\delta^* W=0$. Since $W(\bullet)\cong \mathcal V(\bullet)^*$ in finite dimension (Proposition \ref{prop:wit-separate-span}), the set of such $W$ is exactly $\ker\delta^*$ inside $C^k(\mathcal U,\mathcal V)^*$. Thus $c\in(\ker\delta^*)^\perp=S$, i.e. $c=\delta b$ and $[c]=0$.
\end{proof}

In finite dimension, the linear span of the dual separable cone equals the whole space, hence $W(U)=\text{span}_{\mathbb R}\mathsf{Wit}(U)=\mathcal V(U)$ (Proposition~\ref{prop:wit-separate-span}). Therefore Čech complexes with coefficients in $\mathcal V$ or $W$ are simultaneously contractible via the same homotopy. We use witnesses only as separating functionals via pairings (Theorem~\ref{thm:witness-vanish}).

\begin{proposition}
\label{prop:farkas}
Let $\sigma=(\sigma_i)_{i\in\mathcal I}\in C^0(\mathcal U,\mathcal D)$ be compatible ($\delta\sigma=0$). Consider the feasibility set
\[
\mathcal F(\sigma):=\{\rho\in\mathcal D(I)\ :\ \Tr_{I\setminus U_i}\rho=\sigma_i\ \ \forall i\}.
\]
Then $\mathcal F(\sigma)=\emptyset$ if and only if there exists a (not necessarily compatible) family $W=(W_i)_{i\in\mathcal I}\in C^0(\mathcal U,W(\bullet))$ such that
\[
\sum_{i}\Tr(W_i\sigma_i)\ <\ \lambda_{\min}\!\Big(\sum_{i} r_i^{*}(W_i)\Big), \qquad r_i^\ast(W_i):=W_i\otimes \mathbf 1_{I\setminus U_i}.
\]
Equivalently, a hyperplane defined by witnesses separates the affine constraint from the PSD cone. If one restricts $W_i$ to the dual separable cone, the same gives a certificate against separable realizations.
\end{proposition}
\begin{proof}
Let $R:\mathcal V(I)\to C^0(\mathcal U,\mathcal V)$ be the linear map $R(\rho)=(r_i(\rho))_i$ given by partial traces $r_i$. The feasible set of marginals realized by density operators is
\[
\mathcal S:=R\bigl(\mathcal D(I)\bigr)\subset C^0(\mathcal U,\mathcal V),
\]
a compact convex set since $\mathcal D(I)$ is compact convex and $R$ is linear and continuous.

($\Rightarrow$) Assume $F(\sigma)=\varnothing$, i.e.\ $\sigma\notin\mathcal S$.
By the strict separation theorem for closed convex sets, there exist a covector $W=(W_i)_i\in C^0(\mathcal U,\mathcal V)^*$ (which we identify with $C^0(\mathcal U,\mathcal V)$ via the trace pairing) and a real number $\beta$ such that
\[
\langle W,s\rangle \ge \beta \quad(\forall s\in\mathcal S),\qquad \langle W,\sigma\rangle<\beta .
\]
Write $A:=\sum_i r_i^*(W_i)\in \mathcal V(I)$ so that, by adjointness of $r_i$ and $r_i^*$,
\[
\inf_{\rho\in \mathcal D(I)}\ \sum_i \Tr\bigl(W_i\,r_i(\rho)\bigr) \;=\;\inf_{\rho\in \mathcal D(I)}\ \Tr\bigl(A\rho\bigr) \;=\;\lambda_{\min}(A).
\]
Since $R(\rho)\in\mathcal S$ for all $\rho\in\mathcal D(I)$, we may choose $\beta=\lambda_{\min}(A)$, and the separation inequality becomes
\[
\sum_i \Tr(W_i\,\sigma_i)\;<\;\lambda_{\min}\!\Bigl(\sum_i r_i^*(W_i)\Bigr),
\]
which is exactly the desired inequality.

($\Leftarrow$) Conversely, suppose there exist $W_i$ with
\[
\sum_i \Tr(W_i\,\sigma_i)\;<\;\lambda_{\min}\!\Bigl(\sum_i r_i^*(W_i)\Bigr).
\]
If $\rho\in \mathcal D(I)$ realized $\sigma$, then by adjointness
\[
\sum_i \Tr(W_i\,\sigma_i)=\sum_i \Tr\bigl(W_i\,r_i(\rho)\bigr) =\Tr\Bigl(\sum_i r_i^*(W_i)\,\rho\Bigr)\ \ge\ \lambda_{\min}\!\Bigl(\sum_i r_i^*(W_i)\Bigr),
\]
which is a contradiction. Hence $F(\sigma)=\varnothing$.

Finally, if in addition $W_i\in C^*_{\mathrm{sep}}(U_i)$ for all $i$, then $A=\sum_i r_i^*(W_i)\in C^*_{\mathrm{sep}}(I)$, so $\Tr(A\,\rho_{\mathrm{sep}})\ge 0$ for every $\rho_{\mathrm{sep}}\in \mathcal D_{\mathrm{sep}}(I)$. Therefore the stricter condition $\sum_i \Tr(W_i\sigma_i)<0$ rules out separable realizations, reproducing the separable certificate as stated.
\end{proof}

\subsection{\label{sec:pure}Obstructions in pure states}
For a finite index set $I$ and finite subsets $A,B\subseteq I$ with $A\cap B=\varnothing$, we write $A\,|\,B$ to denote the bipartition of $A\cup B$ and the corresponding tensor factorization $H_{A\cup B}\cong H_A\otimes H_B$. For a single subset $U\subseteq I$ we abbreviate $U\,|\,I\setminus U$ for the cut of $I$ into $U$ and its complement. If $A\subseteq B$, the shorthand $A\,|\,B$ means the internal cut of $B$ into $A$ and $B\setminus A$, i.e.\ $H_B\cong H_A\otimes H_{B\setminus A}$. In particular, with $U_{ip}:=U_i\cap U_p$, the expression $U_{ip}\,|\,U_p$ stands for the bipartition $H_{U_p}\cong H_{U_{ip}}\otimes H_{U_p\setminus U_{ip}}$.

\begin{lemma}
\label{lem:pure-ordered}
Let $\mathcal P(U)\subset\mathcal D(U)$ denote the presheaf of pure states. Fix a finite cover $\mathcal U=\{U_i\}_{i\in\mathcal I}$ and a family $\{\rho_i\in\mathcal P(U_i)\}_{i\in\mathcal I}$ such that for all $i,j$ the overlap marginals $\rho_i|_{U_{ij}}$ and $\rho_j|_{U_{ij}}$ coincide and are pure. Choose unit vectors $\psi_i\in H_{U_i}$ with $\rho_i=|\psi_i\rangle\!\langle\psi_i|$. On overlaps $U_{ij}$ write
\(
\psi_i=\xi_{ij}\otimes\chi_i^{(ij)},
\ \psi_j=g_{ij}\,\xi_{ij}\otimes\chi_j^{(ij)},
\)
with $g_{ij}\in U(1)$ and unit vectors $\chi_\bullet^{(ij)}$ on the complementary factors. Then $g=\{g_{ij}\}\in Z^1(\mathcal U,U(1))$ is a $U(1)$ Čech $1$–cocycle, and the following are equivalent:
\begin{enumerate}[label=(\roman*)]
\item\label{it:pure-glue}
There exists $\rho=|\Psi\rangle\!\langle\Psi|\in\mathcal P(I)$ with $\rho|_{U_i}=\rho_i$ for all $i$.
\item\label{it:pure-cocycle}
$[g]=0$ in $H^1(\mathcal U,U(1))$, i.e. $g_{ij}=e^{i(\alpha_j-\alpha_i)}$ on overlaps for some phases $\{\alpha_i\}$.
\end{enumerate}
If the cover contains (or refines to) all singletons $\{j\}$, then \ref{it:pure-glue} is also equivalent to:
\begin{enumerate}[label=(\roman*),resume]
\item\label{it:pure-sep}
The global pure state is fully separable, $\Psi=\bigotimes_{j\in I}\phi_j$.
\end{enumerate}
\end{lemma}

\begin{proof}
On triple overlaps one computes $g_{ij}g_{jk}g_{ki}=1$, so $g$ is a $U(1)$–valued cocycle.

\emph{\ref{it:pure-glue}$\Rightarrow$\ref{it:pure-cocycle}.}
If $\rho=|\Psi\rangle\!\langle\Psi|$ restricts to $\rho_i$, then $\Psi=\psi_i\otimes\phi_i$ across $U_i\,|\,I\setminus U_i$, so $\psi_i|_{U_{ij}}=\psi_j|_{U_{ij}}$ as vectors; hence we can choose representatives with $g_{ij}=1$ and $[g]=0$.

\emph{\ref{it:pure-cocycle}$\Rightarrow$\ref{it:pure-glue}.}
Pick a spanning tree $T$ in the nerve of $\mathcal U$ and choose phases $\alpha_i$ so that $g_{ij}=e^{i(\alpha_j-\alpha_i)}=1$ on every tree edge $(i,j)\in T$. Replace $\psi_i$ by $\psi'_i:=e^{i\alpha_i}\psi_i$. Then on each edge $(i,j)\in T$ we have $\psi'_i|_{U_{ij}}=\psi'_j|_{U_{ij}}=\xi_{ij}$ as vectors. Glue inductively along the tree: pick a root $i_1$ and set $\Psi_{(1)}:=\psi'_{i_1}$. When attaching a new vertex $i$ with parent $p$ in $T$, write $\psi'_p=\xi_{ip}\otimes\eta_{p}$ and $\psi'_i=\xi_{ip}\otimes\eta_i$ across $U_{ip}\,|\,U_p\!\setminus\!U_{ip}$ and $U_{ip}\,|\,U_i\!\setminus\!U_{ip}$, and put $\Psi_{(\text{new})}:=(\Psi_{(\text{old})}|_{U_p})$ extended by $\eta_i$ on $U_i\setminus U_{ip}$. This produces a vector on the union whose restriction to each already glued $U_j$ is still $\psi'_j$, and whose restriction to $U_i$ is $\psi'_i$. Since $T$ has no cycles and pairwise overlaps are already matched along edges, no further phase conditions arise. Consistency on overlaps with non-parent neighbours follows along the unique tree path (the cocycle is trivial on cycles). Continuing yields $\Psi\in H_I$ with $\Psi|_{U_i}=\psi'_i$, hence $\rho|_{U_i}=\rho_i$.

If the cover contains singletons, then purity of all single‑site marginals forces $\Psi$ to have Schmidt rank $1$ across every $j\,|\,I\setminus\{j\}$, hence $\Psi=\bigotimes_{j\in I}\phi_j$.
\end{proof}

The argument above is independent of the ordered quotient and uses only the projective nature of pure states. In the ordered framework one may regard each $\rho_i$ as a class in $\mathcal V(U_i)$ represented by a rank–one projector. The proof above shows that the only obstruction to a global pure extension is the $U(1)$ Čech class $[g]$, not an ordered entanglement obstruction. When the cover contains singletons, any global pure extension is necessarily fully separable, so no entanglement–driven obstruction arises in the pure case beyond the phase cocycle.

\subsection{\label{sec:local_entanglement}Local entanglement}
Fix once and for all a finite–dimensional Hilbert space $H_{\mathrm{aux}}$ with dimension $d\ge2$ and a faithful state $\tau_{\mathrm{aux}}\in\mathrm{Dens}(H_{\mathrm{aux}})$. Let $\mathrm{Aux}:=\{\mathsf a_1,\mathsf a_2,\ldots\}$ be a countable set of formal labels
disjoint from the physical index set $I$. For each $\ell\ge1$ we identify $H_{\mathsf a_\ell}\cong H_{\mathrm{aux}}$, and write $\mathbf 1_{\mathsf a_\ell}$ and $\tau_{\mathsf a_\ell}$ for the corresponding identity and faithful state. For $q\ge0$ put $A_q:=\{\mathsf a_1,\ldots,\mathsf a_q\}$ (with $A_0=\varnothing$).

Given a finite region $S\subseteq I$, its $q$-fold ancilla–cosimplicial thickening is the disjoint union
\[
   S^{(q)} \;:=\; S\ \sqcup\ A_q, \qquad H_{S^{(q)}} \;=\; H_S \otimes H_{\mathrm{aux}}^{\otimes q}.
\]
We use the same auxiliary labels $A_q$ for all $S$ in a fixed construction, so that restriction (partial trace) maps ignore the auxiliary legs and remain compatible across overlaps. As before, let $\mathcal V(S)\ :=\ \mathrm{Herm}\big(H_S\big)$ denote the real vector space of Hermitian operators on $H_S$.

For a finite cover $\mathcal U=\{U_i\}_{i\in\mathcal I}$ and $p,q\ge0$ set
\[
  C^{p,q}(\mathcal U) \;:=\; \prod_{i_0<\cdots<i_p} \mathcal V\big( (U_{i_0}\cap\cdots\cap U_{i_p})^{(q)} \big),
\]
where the thickening $(\cdot)^{(q)}$ uses the fixed $A_q$ above. The differential $\delta_{\mathrm C}:C^{p,q}\!\to C^{p+1,q}$ is the Čech coboundary, which we used in the previous subsections, built from restrictions to overlaps.

In the following, we refer to complexes as ``ancilla'' (vertical) or ``Čech'' (horizontal) according to their direction in the following diagram:
\[ \begin{tikzcd}[column sep=2.2em,row sep=1.8em,ampersand replacement=\&] C^{p,q+1} \arrow[r,"\delta_{\mathrm C}"] \& C^{p+1,q+1}\\ C^{p,q} \arrow[r,"\delta_{\mathrm C}"] \arrow[u,"\delta_E"] \& C^{p+1,q} \arrow[u,"\delta_E"] \end{tikzcd} \]

\medskip
Let
\[
E:\ \mathrm{Herm}(H\!\otimes\!H_{\mathrm{aux}})\longrightarrow \mathrm{Herm}(H\!\otimes\!H_{\mathrm{aux}}), \qquad E(Z):=\operatorname{Tr}_{\mathrm{aux}}(Z)\otimes \tau_{\mathrm{aux}},
\]
so $E$ is completely positive trace-preserving (CPTP) and idempotent ($E^2=E$). Note that $E(\mathbf 1_H\otimes\mathbf 1_{\mathrm{aux}})=d\,\mathbf 1_H\otimes \tau_{\mathrm{aux}}$ in general, hence $E(\mathbf 1)\neq \mathbf 1$ unless $\tau_{\mathrm{aux}}=\mathbf 1/d$.

For $X\in \mathcal V\big(S^{(q)}\big)$ and $q\ge 0$ define
\begin{align*}
d^{(0)}_q(X) &:= X\otimes \tau_{\mathsf a_{q+1}},\\
d^{(1)}_q(X) &:= \bigl(\mathrm{id}^{\otimes q}\otimes E\bigr)\bigl(X\otimes \tau_{\mathsf a_{q+1}}\bigr)\qquad\text{(reset the newly added slot)},\\
d^{(i)}_q(X) &:= \bigl(\mathrm{id}^{\otimes (i-2)}\otimes E\otimes \mathrm{id}^{\otimes (q-i+2)}\bigr)\bigl(X\otimes \tau_{\mathsf a_{q+1}}\bigr),
\qquad 2\le i\le q+1,
\end{align*}
where for $i\ge 2$ the map $E$ acts on the $(i\!-\!1)$‑st pre‑existing auxiliary leg. Set $\displaystyle \delta_E:=\sum_{m=0}^{q+1}(-1)^m d^{(m)}_q$ and extend componentwise to $C^{p,q}(\mathcal U)$.

\begin{lemma}
\label{lem:ancilla-cosimplicial}
For all $q\ge 0$ and $0\le i<j\le q+2$,
\[
d^{(j)}_{q+1}\circ d^{(i)}_q \;=\; d^{(i)}_{q+1}\circ d^{(j-1)}_q.
\]
Consequently $\delta_E^2=0$.
\end{lemma}
\begin{proof}
Let $J_q:\mathcal V(S^{(q)})\to\mathcal V(S^{(q+1)})$ be state insertion, $J_q(X)=X\otimes\tau$. 
On $\mathcal V(S^{(r)})$, we write
\[
N_{r}:=\mathrm{id}^{\otimes (r-1)}\otimes E\qquad\text{(reset the ``new'' slot)},
\quad
R^{(m)}_{r}:=\mathrm{id}^{\otimes (m-1)}\otimes E\otimes \mathrm{id}^{\otimes (r-m)}
\]
for $1\le m\le r$ (reset the $m$‑th pre‑existing slot). Then
\[
d^{(0)}_q=J_q,\qquad d^{(1)}_q=N_{q+1}\circ J_q,\qquad
d^{(i)}_q=R^{(i-1)}_{q+1}\circ J_q\ \ (i\ge2).
\]

We use the following identities, valid for any faithful $\tau$:
\begin{align*}
\text{(i)}\ & R^{(m)}_{q+1}\circ J_q \;=\; J_q\circ R^{(m)}_{q} \quad(1\le m\le q),\\
\text{(ii)}\ & N_{q+1}\circ J_q \;=\; J_q \qquad\qquad\qquad\quad\ \ \ (\text{since }E(\tau)=\tau,\ \mathrm{Tr}\,\tau=1),\\
\text{(iii)}\ & J_{q+1}\circ N_{q+1} \;=\; R^{(q+1)}_{q+2}\circ J_{q+1}.
\end{align*}
Identity (i) says a reset on an old slot commutes with inserting a fresh state; (ii) says resetting the newly inserted slot does nothing; (iii) says that after the next insertion, the previously new slot becomes the $(q\!+\!1)$‑st old slot.

Using (i)–(iii) one rewrites both sides of $d^{(j)}_{q+1}d^{(i)}_q=d^{(i)}_{q+1}d^{(j-1)}_q$ as the same composition of two resets (possibly on the same leg) followed by $J_{q+1}J_q$. Resets on different legs commute. If the same leg is hit twice, $E^2=E$ applies. The alternating sum of cofaces therefore squares to zero, hence $\delta_E^2=0$.
\end{proof}

$\delta_{\mathrm C}$ acts on only physical legs and each $d^{(m)}_q$ acts only on auxiliary legs, so $\delta_{\mathrm C}\delta_E=\delta_E\delta_{\mathrm C}$. Therefore $(C^{\bullet,\bullet}(\mathcal U),\delta_{\mathrm C},\delta_E)$ is a bicomplex.

\begin{HLblock}
\begin{definition}
\label{def:Eq-ordered-cosimplicial}
For a cover $\mathcal U$ set $C^{0,q}(\mathcal U):=\prod_{i}\mathcal V\big(U_i^{(q)}\big)$ and define
\[
E^{q}(\mathcal U)\ :=\ H^{q}\!\big(C^{0,\bullet}(\mathcal U),\delta_E\big).
\]
We call $E^{q}$ the \emph{local entanglement groups}. Here ``local'' means within a single patch of the chosen cover.
\end{definition}
\end{HLblock}

At degree $q=0$, for any faithful $\tau_{\rm aux}$, we have
\[
d^{(0)}_0(X)=X\otimes \tau_{\rm aux},\qquad d^{(1)}_0(X)=(\mathrm{id}\otimes E)(X\otimes \tau_{\rm aux})=X\otimes \tau_{\rm aux},
\]
hence $\delta_E\equiv 0$ in degree $0$ and $E^0(U)\cong \mathcal V(U)$. Nontriviality at $q=0$ is decided by ordinary separability witnesses on $U$.

\medskip
\noindent\textbf{Ancilla--state independence.}
If $\tau,\tau'$ are faithful and $E_\tau,E_{\tau'}$ are the corresponding reset maps, then after choosing a trace--preserving linear isomorphism
$T:\Herm(H_{\mathrm{aux}})\to\Herm(H_{\mathrm{aux}})$ with $T(\tau)=\tau'$, the induced chain map $F_\bullet=\id\otimes T^{\otimes(\bullet)}$ yields isomorphisms
\[
E^q_\tau(\mathcal U)\;\xrightarrow{\;\sim\;}\;E^q_{\tau'}(\mathcal U)\qquad(\forall\,q\ge 0).
\]
This identification depends on the choice of $T$, but in particular the groups $E^q_\tau(\mathcal U)$ are all (noncanonically) isomorphic for faithful $\tau$.

\medskip
Entanglement already visible on a single patch is decided at degree $q=0$ by ordinary separability witnesses. For $q\ge 1$, the groups $E^q(\mathcal U)=H^q(C^{0,\bullet}(\mathcal U),\delta_E)$ organize auxiliary--traceless components in a cohomological way (and are independent, up to noncanonical isomorphism, of the faithful choice of $\tau_{\mathrm{aux}}$; see Lemma~\ref{lem:tau-independence-cosimplicial}). Operational witness-based statements for the signed stabilize--then--reset pairing are captured instead by LED$(q)$ (Definition~\eqref{def:LED}) and its parity collapse (Proposition~\ref{prop:LED-tau-independence} / Theorem~\ref{thm:completeness-weak}). See also Definition~\ref{def:witness-trivial} for the fact that the signed pairing detects $\delta_E(Y)$ rather than the cohomology class $[Y]$.

Let $\tau_{\rm aux}$ and $\tau'_{\rm aux}$ be faithful states, and let $\delta_E,\delta_{E'}$ be the ancilla differentials built from $E(Z)=\Tr_{\rm aux}(Z)\otimes\tau_{\rm aux}$ and $E'(Z)=\Tr_{\rm aux}(Z)\otimes\tau'_{\rm aux}$. By Lemma~\ref{lem:tau-independence-cosimplicial}
there is a degreewise linear isomorphism $F_\bullet$ with $\delta_{E'}\!\circ F_q = F_{q+1}\!\circ \delta_E$, hence a natural identification
$E^q_{\tau_{\rm aux}}(U)\cong E^q_{\tau'_{\rm aux}}(U)$ for all $q$.

For the convenient choice $\tau_{\rm aux}=\mathbf 1/d$, one has $d^{(0)}_0=d^{(1)}_0$, so $\delta_E=0$ in degree $q=0$ and compatible ancilla–cosimplicial witnesses coincide with all witnesses on the patch. The $q=0$ pairing reduces to the usual separability test on $U$.

\begin{lemma}
\label{lem:tau-independence-cosimplicial}
Let $\tau,\tau'$ be faithful states on $H_{\mathrm{aux}}$. Choose a unital, trace–preserving linear isomorphism
$T:\mathrm{Herm}(H_{\mathrm{aux}})\to \mathrm{Herm}(H_{\mathrm{aux}})$ with $T(\tau)=\tau'$. For $q\ge 0$ define
$F_q:=\mathrm{id}\otimes T^{\otimes q}:\mathcal V(U^{(q)})\to \mathcal V(U^{(q)})$. Then for all $q$ and all cofaces $d^{(m)}_{q,(\bullet)}$ built from $E_{(\bullet)}(Z)=\operatorname{Tr}_{\mathrm{aux}}(Z)\otimes (\bullet)$ and the insertion of the corresponding state,
\[
d^{(m)}_{q,\tau'}\circ F_q \;=\; F_{q+1}\circ d^{(m)}_{q,\tau}\qquad(0\le m\le q+1).
\]
Consequently $F_\bullet$ is a chain isomorphism and induces (noncanonical) isomorphisms $E^q_{\tau}(U)\cong E^q_{\tau'}(U)$ for all $q$.
Moreover, once $T$ is fixed, these isomorphisms are functorial with respect to restrictions.
\end{lemma}

\begin{proof}
We first show the key identity on a single auxiliary factor:
\begin{equation}
\label{eq:one-leg}
E_{\tau'}\circ(\mathrm{id}\otimes T)\;=\;(\mathrm{id}\otimes T)\circ E_{\tau}
\quad\text{on }\mathrm{Herm}(H\otimes H_{\mathrm{aux}}).
\end{equation}
For any $Z\in \mathrm{Herm}(H\otimes H_{\mathrm{aux}})$,
using linearity of the partial trace, trace–preservation of $T$, and $T(\tau)=\tau'$,
\[
\begin{aligned}
E_{\tau'}\!\big((\mathrm{id}\otimes T)(Z)\big)
&= \mathrm{Tr}_{\mathrm{aux}}\!\big((\mathrm{id}\otimes T)(Z)\big)\ \otimes\ \tau'
= \big(\mathrm{id}\otimes \mathrm{Tr}\circ T\big)(Z)\ \otimes\ \tau'\\
&= \big(\mathrm{id}\otimes \mathrm{Tr}\big)(Z)\ \otimes\ \tau'
= \big(\mathrm{id}\otimes T\big)\!\Big(\big(\mathrm{id}\otimes \mathrm{Tr}\big)(Z)\ \otimes\ \tau\Big)\\
&= (\mathrm{id}\otimes T)\!\big(E_{\tau}(Z)\big),
\end{aligned}
\]
which is \eqref{eq:one-leg}.

\medskip
Fix $q\ge 0$ and $X\in \mathcal V(U^{(q)})$. For $m=0$,
\[
d^{(0)}_{q,\tau'}(F_qX)=(F_qX)\otimes \tau'=(\mathrm{id}\otimes T^{\otimes q})(X)\otimes T(\tau)
=F_{q+1}(X\otimes\tau)=F_{q+1}d^{(0)}_{q,\tau}(X).
\]
For $m=1$,
\[
\begin{aligned}
F_{q+1}d^{(1)}_{q,\tau}(X)
&=F_{q+1} \big((\mathrm{id}^{\otimes q}\otimes E_\tau)(X\otimes\tau)\big)\\
&=(\mathrm{id}^{\otimes q}\otimes E_{\tau'})\big((\mathrm{id}\otimes T^{\otimes q})(X)\otimes T(\tau)\big)
= d^{(1)}_{q,\tau'}(F_qX),
\end{aligned}
\]
using $E_{\tau'}\circ(\mathrm{id}\otimes T)=(\mathrm{id}\otimes T)\circ E_\tau$ on the targeted leg. For $2\le m\le q+1$ the same argument on the $m$-th pre-existing leg gives $d^{(m)}_{q,\tau'}\circ F_q=F_{q+1}\circ d^{(m)}_{q,\tau}$. Summing with the alternating signs proves $\delta_{E_{\tau'}}\circ F_q=F_{q+1}\circ \delta_{E_\tau}$. Since $T$ is a linear isomorphism, each $F_q=\mathrm{id}\otimes T^{\otimes q}$ is a linear isomorphism with inverse $\mathrm{id}\otimes (T^{-1})^{\otimes q}$, so $F_\bullet$ is a chain isomorphism. Hence it induces natural isomorphisms on ancilla cohomology: $E^{q}_{\tau}(U)\cong E^{q}_{\tau'}(U)$ for all $q$.

\medskip
The restriction maps are partial traces on physical legs. Since each $F_q$ acts only on auxiliary legs, $F_\bullet$ commutes with all horizontal restriction maps. Therefore $F_\bullet$ respects the full bicomplex and yields natural identifications also after taking Čech cohomology in the horizontal direction. This completes the proof.
\end{proof}

\medskip
Over refinements $\mathcal V\succeq\mathcal U$, use the natural push–forwards to define the cover–independent invariants
\[
 \overline{Q}^{\,0}:=\varinjlim_{\mathcal U} Q^{0}(\mathcal U),\qquad
 \overline{R}^{\,0}:=\varinjlim_{\mathcal U} R^{0}(\mathcal U),\qquad
 \overline{E}^{\,q}:=\varinjlim_{\mathcal U} E^{q}(\mathcal U).
\]

\begin{definition}
\label{def:LED}
Fix a patch $U$ and a faithful auxiliary state $\tau_{\rm aux}$.  For $q\ge1$ set
\[
\mathcal C^q(\rho_U)
:=\ \rho_U\otimes\tau_{\rm aux}^{\otimes q}
\ +\ \mathcal V(U)\otimes \mathcal V_0(H_{\rm aux})^{\otimes q}
\ \subset\ \mathcal V\big(U^{(q)}\big),
\]
where $\mathcal V_0(H_{\rm aux})$ denotes the traceless subspace. We say that $\rho_U$ has LED (Local Entanglement Detectability) of order $q\ge1$ if there exist $Y\in\mathcal C^q(\rho_U)$ and a separability witness $W\in\mathsf{Wit}\big(U^{(q+1)}\big)$
such that
\[
\sum_{m=0}^{q+1}(-1)^m\,\Tr\!\big[W\,d_q^{(m)}(Y)\big]\ <\ 0.
\]
For $q=0$ (where $d^{(0)}_0=d^{(1)}_0$ and hence $\delta_E\equiv0$) we declare $\mathrm{LED}(0)$ if and only if there exists $W\in\mathsf{Wit}(U)$ with $\Tr(W\rho_U)<0$.
\end{definition}

\begin{proposition}
\label{prop:LED-tau-independence}
Fix a patch $U\subset I$, a reduced state $\rho_U\in D(U)$, and a faithful auxiliary state $\tau_{\mathrm{aux}}$. Let $q\ge 1$, and let $\delta_E=\sum_{m=0}^{q+1}(-1)^m d_q^{(m)}$ be the ancilla differential built from $\tau_{\mathrm{aux}}$.

\begin{enumerate}
\item If $q$ is even, then for every $Y\in C_q(\rho_U)$ one has $\delta_E(Y)=0$, hence
\[
\sum_{m=0}^{q+1}(-1)^m\Tr\!\left(W\,d_q^{(m)}(Y)\right)
=\Tr\!\left(W\,\delta_E(Y)\right)=0
\]
for every Hermitian $W$ on $U^{(q+1)}$. In particular, $\mathrm{LED}(q)$ never occurs for even $q$.

\item If $q$ is odd, then $\rho_U$ is entangled (i.e.\ not fully separable on $U$) if and only if $\rho_U$ has $\mathrm{LED}(q)$. Moreover this equivalence is independent of the choice of faithful $\tau_{\mathrm{aux}}$.

More explicitly, if $W_U\in\mathrm{Wit}(U)$ satisfies $\Tr(W_U\rho_U)<0$, then for any faithful $\tau_{\mathrm{aux}}$ the choices
\[
Y:=\rho_U\otimes \tau_{\mathrm{aux}}^{\otimes q}\in C_q(\rho_U),
\qquad
W:=W_U\otimes \mathbf 1_{A_{q+1}}\in \mathrm{Wit}\!\left(U^{(q+1)}\right)
\]
yield
\[
\sum_{m=0}^{q+1}(-1)^m\Tr\!\left(W\,d_q^{(m)}(Y)\right) =\Tr(W_U\rho_U)<0.
\]
\end{enumerate}
\end{proposition}
\begin{proof}
Write $\tau:=\tau_{\mathrm{aux}}$.
Every $Y\in C_q(\rho_U)$ decomposes uniquely as
\[
Y=\iota_q(\rho_U)+Y_0,\qquad 
\iota_q(\rho_U)=\rho_U\otimes \tau^{\otimes q},\qquad 
Y_0\in V(U)\otimes V_0(H_{\mathrm{aux}})^{\otimes q}.
\]
For the reindexed cofaces, resetting the freshly inserted state does nothing, hence $d_q^{(0)}=d_q^{(1)}$ (because $E(\tau)=\tau$). Therefore the $m=0,1$ contributions cancel in \(\delta_E=\sum_{m=0}^{q+1}(-1)^m d_q^{(m)}\).

For $m\ge 2$, the coface $d_q^{(m)}$ applies $E$ to the $(m-1)$-st \emph{old} auxiliary leg. Since each auxiliary tensorand of $Y_0$ is traceless, applying $E$ to any old auxiliary leg kills $Y_0$, so $d_q^{(m)}(Y_0)=0$ for all $m\ge 2$. On the other hand, on $\iota_q(\rho_U)$ every auxiliary leg is already $\tau$, so each reset leaves it unchanged and
\[
d_q^{(m)}\bigl(\iota_q(\rho_U)\bigr)=\iota_{q+1}(\rho_U)=\rho_U\otimes\tau^{\otimes(q+1)}
\qquad (m\ge 2).
\]
Hence
\[
\delta_E(Y)=\Bigl(\sum_{m=2}^{q+1}(-1)^m\Bigr)\,\iota_{q+1}(\rho_U)
=
\begin{cases}
0, & q \text{ even},\\
\iota_{q+1}(\rho_U), & q \text{ odd}.
\end{cases}
\]

If $q$ is even, then $\delta_E(Y)=0$ for all $Y\in C_q(\rho_U)$, so for every Hermitian $W$ on $U^{(q+1)}$,
\[
\sum_{m=0}^{q+1}(-1)^m \Tr\!\bigl(W d_q^{(m)}(Y)\bigr)=\Tr\!\bigl(W\,\delta_E(Y)\bigr)=0.
\]
In particular LED$(q)$ cannot occur.

Assume $q$ is odd. If LED$(q)$ holds, there exist $Y\in C_q(\rho_U)$ and $W\in \mathrm{Wit}\bigl(U^{(q+1)}\bigr)$ such that
\[
0>\sum_{m=0}^{q+1}(-1)^m \Tr\!\bigl(W d_q^{(m)}(Y)\bigr)
=\Tr\!\bigl(W(\rho_U\otimes\tau^{\otimes(q+1)})\bigr).
\]
Define the witness restriction
\[
W_U:=\Tr_{A_{q+1}}\!\bigl(W(\mathbf 1_U\otimes\tau^{\otimes(q+1)})\bigr)\in V(U).
\]
For every fully separable $\sigma$ on $U$, the state $\sigma\otimes\tau^{\otimes(q+1)}$ is fully separable on $U^{(q+1)}$, so \(\Tr(W_U\sigma)=\Tr\bigl(W(\sigma\otimes\tau^{\otimes(q+1)})\bigr)\ge 0\). Thus $W_U\in \mathrm{Wit}(U)$, and moreover
\[
\Tr(W_U\rho_U)=\Tr\!\bigl(W(\rho_U\otimes\tau^{\otimes(q+1)})\bigr)<0,
\]
so $\rho_U$ is entangled.

Conversely, if $\rho_U$ is entangled, choose $W_U\in \mathrm{Wit}(U)$ with $\Tr(W_U\rho_U)<0$ and set
$Y:=\rho_U\otimes\tau^{\otimes q}$ and $W:=W_U\otimes \mathbf 1_{A_{q+1}}$.
Then $W\in\mathrm{Wit}\bigl(U^{(q+1)}\bigr)$ and
\[
\sum_{m=0}^{q+1}(-1)^m \Tr\!\bigl(W d_q^{(m)}(Y)\bigr)
=\Tr(W_U\rho_U)<0,
\]
so LED$(q)$ holds.

The odd-degree equivalence depends only on entanglement of $\rho_U$, hence it is independent of the faithful choice of $\tau_{\mathrm{aux}}$.
\end{proof}

Let \(\tau_{\mathrm{aux}}=\mathbf 1/d\) on \(H_{\mathrm{aux}}\). Then, in degree \(q=0\),
\[
d^{(0)}_0(X)=X\otimes \mathbf 1,\qquad
d^{(1)}_0(X)=E(X\otimes \mathbf 1)=X\otimes \mathbf 1,
\]
so \(\delta_{\mathrm E}=d^{(0)}_0-d^{(1)}_0= 0\) on \(\mathcal V(U^{(0)})\). Consequently, as real vector spaces,
\[
E^0(U)=\ker\left(\delta_{\mathrm E}:\mathcal V\left (U^{(0)}\right) \to \mathcal V\left (U^{(1)}\right) \right)\ \cong\ \mathcal V\left (U^{(0)}\right) .
\]
Indeed, \(E(Z)=\mathrm{Tr}_{\mathrm{aux}}(Z)\otimes \tau_{\mathrm{aux}}\) gives \(E(X\otimes\mathbf 1)=\mathrm{Tr}_{\mathrm{aux}}(\mathbf 1)\,X\otimes\tau_{\mathrm{aux}}
=d\,X\otimes(\mathbf 1/d)=X\otimes\mathbf 1\).

At \(q=0\) nontriviality is purely operational (see also Section \ref{sec:example}): a class \([X]\in E^0(U)\) is witness–nontrivial if and only if there exists a separability witness \(W\) on \(U\) with \(\mathrm{Tr}(WX)<0\). Thus, group exactness and witness–(non)triviality are logically distinct notions on this row.

\begin{HLblock}
\begin{theorem}
\label{thm:completeness-weak}
Let $\rho$ be a global state on $I$. If either of the following holds, then a genuine obstruction is present:

\begin{enumerate}
\item\label{R0} There exists a finite cover $\mathcal U$ with
\[
R^0(\mathcal U)\;=\;\ker\bigl(j:\mathcal V(I)\to H^0(\mathcal U,\mathcal V)\bigr)\;\neq\;0.
\]

\item\label{LED_item} (Operational obstruction) There exist a finite cover $\mathcal U$ and a patch $U_i\in\mathcal U$ such that the reduced state $\rho|_{U_i}$ satisfies $\mathrm{LED}(q)$ (Definition~\ref{def:LED}) for some $q\ge 0$. Then $\rho$ is entangled. Moreover, $\mathrm{LED}(q)$ can occur only for odd $q$. For even $q$ the signed stabilize--reset sum vanishes for every $Y\in\mathcal C^q(\rho|_{U_i})$.
\end{enumerate}

\noindent
If, in addition, there exists a representative $\widetilde Y\in\mathcal C^q(\rho|_{U_i})$ with $\delta_E\widetilde Y=0$ and $\widetilde Y\notin \mathrm{im}\,\delta_E$, then the $U_i$--column defines a nonzero class $[\widetilde Y]\in E^q(\mathcal U)$; in particular $\overline{E}^{\,q}:=\varinjlim_{\mathcal U}E^q(\mathcal U)$ is nonzero.
\end{theorem}
\end{HLblock}

\begin{proof}
Statement \eqref{R0} is immediate from the definition of $R^0$.

For \eqref{LED_item}, fix $U:=U_i$ and write $q\ge 1$ (the case $q=0$ is the usual witness test on $U$). By Definition~\ref{def:LED}, there exist $Y\in\mathcal C^q(\rho_U)$ and $W\in\mathsf{Wit}\big(U^{(q+1)}\big)$ with
\[
\sum_{m=0}^{q+1}(-1)^m\,\Tr\!\big[W\,d^{(m)}_q(Y)\big]\;<\;0.
\]
Decompose $Y=\iota_q(\rho_U)+Y_0$ with $Y_0\in \mathcal V(U)\otimes \mathcal V_0(H_{\rm aux})^{\otimes q}$. Using $d^{(0)}_q=d^{(1)}_q$ (state insertion followed by reset on the new slot) and that each $d^{(m)}_q$, $m\ge 2$, kills the traceless part on the $m{-}1$‑st old auxiliary leg, one obtains
\[
\delta_E(Y)\;=\;\sum_{m=2}^{q+1}(-1)^m\,d^{(m)}_q\big(\iota_q(\rho_U)\big)
\;=\;\Big(\sum_{m=2}^{q+1}(-1)^m\Big)\,\iota_{q+1}(\rho_U)
\;=\;
\begin{cases}
0, & q\ \text{even},\\[2pt]
\iota_{q+1}(\rho_U), & q\ \text{odd}.
\end{cases}
\]
Hence, for even $q$ the signed combination vanishes for every $Y\in\mathcal C^q(\rho_U)$, so LED$(q)$ cannot occur.

Assume $q$ is odd. Then $\delta_E(Y)=\iota_{q+1}(\rho_U)$ and the LED inequality reduces to
\[
\Tr\!\big[W\,(\rho_U\otimes\tau_{\rm aux}^{\otimes(q+1)})\big]\;<\;0.
\]
Define the (functorial) witness restriction
\[
W_U\ :=\ r^{U^{(q+1)}}_{U}(W)
\ =\ \Tr_{A_{q+1}}\!\big[\,W\,(\mathbf 1_U\otimes \tau_{\rm aux}^{\otimes(q+1)})\big]
\ \in\ \mathsf{Wit}(U).
\]
For any separable $\sigma$ on $U$, the state $\sigma\otimes\tau_{\rm aux}^{\otimes(q+1)}$ is separable on $U^{(q+1)}$, hence
\(
\Tr(W_U\sigma)=\Tr\!\big[W(\sigma\otimes\tau_{\rm aux}^{\otimes(q+1)})\big]\ge0
\),
so $W_U$ is indeed a separability witness on $U$. By adjointness of $r$ and $\iota$,
\[
\Tr(W_U\rho_U)
=\Tr\!\big[W\,(\rho_U\otimes\tau_{\rm aux}^{\otimes(q+1)})\big]\;<\;0,
\]
which shows that $\rho_U$ is entangled. Since every globally separable state has separable marginals, the global state $\rho$ is entangled as well.

Then the last statement of the theorem follows by definition. 
\end{proof}

\vskip0.3cm
Bell tests are a standard tool for certifying entanglement (as highlighted by the 2022 Nobel Prize to Clauser, Aspect, and Zeilinger). It provides a natural $q=0$ separability witness on a two‑site patch (Proposition~\ref{prop:CHSH-q0}). As shown in Proposition~\ref{prop:CHSH-no-new-q>0}, their ancilla extensions either vanish (even $q$) or collapse to the same expectation (odd $q$), offering no additional power for $q>0$. This motivates the LED$(q)$ machinery introduced in this work.

\begin{proposition}\label{prop:CHSH-q0}
Let $U=\{i,j\}$ be a two-site patch and let $A_0,A_1$ (on $i$) and $B_0,B_1$ (on $j$) be dichotomic observables with spectrum $\{\pm1\}$. Define the CHSH operator
\[
\mathcal{B}\;=\;A_0\otimes(B_0+B_1)\;+\;A_1\otimes(B_0-B_1),
\qquad
W_{\mathrm{CHSH}}\;:=\;2\,\mathbf{1}-\mathcal{B}.
\]
Then $W_{\mathrm{CHSH}}\in C^*_{\!\mathrm{sep}}(U)$, i.e. $\Tr[W_{\mathrm{CHSH}}\sigma]\ge 0$ for every separable
$\sigma$ on $U$. Consequently,
\[
\Tr\!\left[W_{\mathrm{CHSH}}\,\rho|_U\right]<0
\quad\Longleftrightarrow\quad
\langle \mathcal{B}\rangle_{\rho|_U}>2
\]
certifies local entanglement on $U$ at degree $q=0$ in our framework.
\end{proposition}

\begin{proof}
For separable $\sigma=\sum_k p_k\,\alpha_k\otimes\beta_k$ with $\alpha_k,\beta_k$ single-site states, the CHSH inequality gives $\langle\mathcal{B}\rangle_\sigma\le 2$. Hence $\Tr[(2\mathbf{1}-\mathcal{B})\sigma]\ge 0$, so $W_{\mathrm{CHSH}}\in C^*_{\!\mathrm{sep}}(U)$. The stated equivalence is immediate.
\end{proof}

\begin{remark}\label{rem:no-converse}
Failure to violate CHSH does not imply separability of $\rho|_U$. Within a fixed patch one may still find other $q=0$ witnesses with negative expectation. Moreover, Proposition~\ref{prop:LED-tau-independence} shows that for odd $q\ge 1$, LED$(q)$ reduces to such a $q=0$ witness test on the same patch, while for even $q\ge 1$ it is impossible.
\end{remark}

\begin{proposition}\label{prop:CHSH-no-new-q>0}
Fix a patch $U$ and $q\ge1$. Consider the ancilla–cosimplicial thickening $U^{(q+1)}$ and the extended witness $\widetilde{W}:=W_{\mathrm{CHSH}}\otimes\mathbf{1}_{\mathrm{anc}}$ on $U^{(q+1)}$. Then:
\begin{enumerate}
\item[(a)] (\emph{Even $q$}) For every $Y\in C_q(\rho_U)$,
\[
\sum_{m=0}^{q+1}(-1)^m\,\Tr\!\big[\widetilde{W}\,d^{(m)}_q(Y)\big] \;=\;0.
\]
Thus no LED$(q)$ detection is possible at even degrees.
\item[(b)] (\emph{Odd $q$}) For every $Y\in C_q(\rho_U)$,
\[
\sum_{m=0}^{q+1}(-1)^m\,\Tr\!\big[\widetilde{W}\,d^{(m)}_q(Y)\big]
\;=\;\Tr\!\left[W_{\mathrm{CHSH}}\,\rho_U\right].
\]
Equivalently, LED$(q)$ with $\widetilde{W}$ holds iff the $q=0$ CHSH test holds on $U$.
\end{enumerate}

In particular, the Bell/CHSH test provides no additional detection power for any $q>0$.
\end{proposition}

\begin{proof}
By the identities of the ancilla differential, for any $Y\in C_q(\rho_U)$ one has $\delta_E(Y)=0$ for even $q$, and $\delta_E(Y)=\iota_{q+1}(\rho_U)$ for odd $q$. Using the
functorial witness restriction
\[
r^{U^{(q+1)}}_{U}(\widetilde{W})=\Tr_{\mathrm{anc}}\!\left[\widetilde{W}\,(\mathbf{1}_U\otimes\tau_{\mathrm{aux}}^{\otimes(q+1)})\right]
= W_{\mathrm{CHSH}},
\]
the signed stabilize–reset pairing reduces to $0$ in case (a) and to $\Tr[W_{\mathrm{CHSH}}\rho_U]$ in case (b). This gives the two claims.
\end{proof}

Consequently, CHSH provides no additional detection power under ancilla thickening: for even $q$ the signed stabilize--reset pairing vanishes, and for odd $q$ it reduces to the $q=0$ CHSH expectation value. More generally, Proposition~\ref{prop:LED-tau-independence} shows the same parity collapse for LED$(q)$ with arbitrary separability witnesses: odd $q\ge 1$ detection is equivalent to an ordinary
$q=0$ witness test on the same patch, while even $q\ge 1$ detection is impossible.

\subsection{\label{sec:example}Examples of local entanglement groups}
Assume throughout the reindexed cofaces of Section \ref{sec:local_entanglement}: for $X\in\mathcal V(U^{(q)})$
\[
d^{(0)}_q(X)=X\otimes \tau_{\mathsf a_{q+1}},~
d^{(1)}_q(X)=(\mathrm{id}^{\otimes q}\otimes E)\bigl(X\otimes \tau_{\mathsf a_{q+1}}\bigr),~
d^{(i)}_q(X)=(\mathrm{id}^{\otimes(i-2)}\otimes E\otimes \mathrm{id}^{\otimes(q-i+2)})\bigl(X\otimes \tau_{\mathsf a_{q+1}}\bigr)
\]
for $2\le i\le q+1$, where $d^{(1)}_q$ resets the new slot and $d^{(i)}_q$ with $i\ge2$ resets the pre‑existing slot $(i{-}1)$. Let $\delta_E=\sum_{m=0}^{q+1}(-1)^m d^{(m)}_q$.

\begin{definition}\label{def:witness-trivial}
Let $U$ be a single patch and consider the ancilla column
\[
V\!\left(U^{(0)}\right)\xrightarrow{\ \delta_E\ }V\!\left(U^{(1)}\right)\xrightarrow{\ \delta_E\ }\cdots .
\]
At degree $q$ we distinguish the following two notions.

\begin{enumerate}
\item \textbf{Cohomological nontriviality.}
A class $[Y]\in E^q(U)=H^q\!\left(V(U^{(\bullet)}),\delta_E\right)$ is nonzero if $Y\in V(U^{(q)})$ satisfies $\delta_E Y=0$ and $Y\notin \mathrm{im}(\delta_E)$.

\item \textbf{LED pairing (obstruction detection).}
For $Y\in V(U^{(q)})$ and any Hermitian operator $W$ on $U^{(q+1)}$ one has the identity
\[
\sum_{m=0}^{q+1}(-1)^m\Tr\!\left(W\,d_q^{(m)}(Y)\right)
=\Tr\!\left(W\,\delta_E(Y)\right).
\]
In particular, this signed pairing vanishes whenever $Y$ is $\delta_E$--closed.
Thus the signed ``stabilize--then--reset'' pairing detects the \emph{obstruction cocycle} $\delta_E(Y)$ (as in $\mathrm{LED}(q)$),
and it does \emph{not} detect the cohomology class $[Y]\in E^q(U)$.
\end{enumerate}
\end{definition}

We first consider a basic \(q{=}1\) identity. Choose an auxiliary operator basis $\{S_b\}_{b\ge0}$ with $S_0:=\tau_{\mathrm{aux}}$ and traceless $S_b$ for $b\ge1$. Then $E(S_0)=S_0$ and $E(S_b)=0$ ($b\ge1$). Every $Y\in\mathcal V(U^{(1)})$ decomposes uniquely as
\[
Y \;=\; A\otimes S_0 \;+\; \sum_{b\ge 1} B_b\otimes S_b,\qquad A,B_b\in\mathcal V(U).
\]
With these cofaces $d^{(0)}_1=d^{(1)}_1$ and $d^{(2)}_1=(E_1\!\cdot\,)\otimes\mathrm{id}$, hence
\begin{equation}\label{eq:q1-identity}
\delta_E(Y) \;=\; (E_1Y)\otimes S_0 \;=\; (A\otimes S_0)\otimes S_0.
\end{equation}
Here $E_1$ means $E$ acting on the pre‑existing auxiliary leg ($d^{(2)}_1=E_1\otimes \mathrm{id}$). Thus $Y$ is $\delta_E$–closed if and only if its old slot is entirely traceless, i.e.\ $A=0$. Since $\delta_E\equiv 0$ in degree $0$, every closed $Y$ at $q=1$ represents a nonzero class in $E^1(U)$.

\begin{proposition}\label{prop:examples}
Let $\tau_{\mathrm{aux}}=\mathbf 1/d$ and define cofaces as above.

\begin{enumerate}
\item[\textup{(Bell)}] For the Bell state $\ket{\Psi^\pm}=\frac{\ket{10}\pm\ket{01}}{\sqrt{2}},\ket{\Phi^\pm}=\frac{\ket{00}\pm\ket{11}}{\sqrt{2}}$ on $S=\{a,b\}$ one has $d^{(0)}_0=d^{(1)}_0$, hence $\delta_E\equiv 0$ on $\mathcal V(S^{(0)})$ and $E^0(S)\cong \mathcal V(S)$. Moreover $[\rho_{\Psi^\pm}],[\rho_{\Phi^\pm}]$ are witness–nontrivial at $q=0$.

\item[\textup{($W$‑state)}] Let $\ket{W_N}=\sum_{j=1}^N\frac{\ket{0\cdots010\cdots0}}{\sqrt{N}}$ be the $W_N$-state (i.e., the uniform superposition of all basis states with a single excitation at $j=1,\cdots,N$). For $S\subset I=\{1,\cdots,N\}$, the marginal state of $\ket{W_N}$ on $S$ is defined by the reduced density matrix $\rho_S=\Tr_{I\setminus S}\ket{W_N}\bra{W_N}$. In particular, for $|S|=2$, $\rho_S=\frac{N-2}{N}\ket{00}\bra{00}+\frac{2}{N}\ket{\Psi^+}\bra{\Psi^+}$, which is entangled for all $N\ge3$, hence is witness‑nontrivial at degree $q=0$. 
\item[\textup{(GHZ\(_{k+1}\))}] Let $U=\{1,\dots,k+1\}$ and let $\rho_{\mathrm{GHZ}}=|\mathrm{GHZ}_{k+1}\rangle\langle \mathrm{GHZ}_{k+1}|$.
Then:
\begin{enumerate}
\item $\rho_{\mathrm{GHZ}}$ is entangled on $U$, hence it is detected already at $q=0$ by some witness $W_U\in\mathrm{Wit}(U)$:
$\Tr(W_U\rho_{\mathrm{GHZ}})<0$.
\item For every odd $q\ge 1$ and every faithful $\tau_{\mathrm{aux}}$, $\rho_{\mathrm{GHZ}}$ satisfies $\mathrm{LED}(q)$ on $U$.
One explicit choice is
$Y=\rho_{\mathrm{GHZ}}\otimes \tau_{\mathrm{aux}}^{\otimes q}$ and $W=W_U\otimes \mathbf 1_{A_{q+1}}$, which gives
\[
\sum_{m=0}^{q+1}(-1)^m\Tr\!\left(W\,d_q^{(m)}(Y)\right)=\Tr(W_U\rho_{\mathrm{GHZ}})<0.
\]
\item For even $q$, $\mathrm{LED}(q)$ never occurs (Proposition~\ref{prop:LED-tau-independence}).
\end{enumerate}
\end{enumerate}
\end{proposition}

\begin{proof}
(Bell \& $W$ states) At $q=0$ one has $d^{(0)}_0(X)=X\otimes \tau_{\rm aux}$ and 
$d^{(1)}_0(X)=(\mathrm{id}\otimes E)(X\otimes \tau_{\rm aux})=X\otimes \tau_{\rm aux}$, so $\delta_E\equiv0$. Entangled two–qubit states admit a separable witness with negative expectation, hence detection at $q=0$.

(GHZ$_3$ case \(k{=}2\)). Take $U=\{1,2,3\}$ and write
\(
Y=\sum_{b\ge1}B_b\otimes S_b\in\mathcal V(U^{(1)})
\)
built linearly from $\rho_{\mathrm{coh}}$ so that the old slot is purely traceless. By \eqref{eq:q1-identity} $\delta_E(Y)=0$, and since $\delta_E\equiv0$ at $q=0$, $[Y]\neq 0\in E^1(U)$. For operational detection, set $Y_{\mathrm{op}}:=\rho_{\mathrm{GHZ}}\otimes S_0\in \mathcal V(U^{(1)})$. Then
\[
\delta_E(Y_{\mathrm{op}})=(E_1Y_{\mathrm{op}})\otimes S_0
=\big(\rho_{\mathrm{GHZ}}\otimes S_0\big)\otimes S_0.
\]
Pick a separable GHZ witness $W_U$ on $U$ with $\Tr(W_U\rho_{\mathrm{GHZ}})<0$, and choose any $R\succeq 0$ on the old slot with $\Tr(R S_0)>0$. For $W=W_U\otimes R\otimes S_0$ one gets
\[
\sum_{m=0}^{2}(-1)^m\,\Tr\!\big[W\,d^{(m)}_1(Y_{\mathrm{op}})\big]
=\Tr\big[W\,\delta_E(Y_{\mathrm{op}})\big]
=\Tr(W_U\rho_{\mathrm{GHZ}})\,\Tr(R S_0)\,\Tr(S_0^2)\;<\;0.
\]

(GHZ$_{k+1}$). The pure state $\rho_{\mathrm{GHZ}}$ is entangled on $U$, hence admits a separability witness $W_U\in\mathrm{Wit}(U)$ with $\Tr(W_U\rho_{\mathrm{GHZ}})<0$. For odd $q\ge 1$, Proposition~\ref{prop:LED-tau-independence} gives LED$(q)$, e.g. with
$Y=\rho_{\mathrm{GHZ}}\otimes\tau_{\mathrm{aux}}^{\otimes q}$ and $W=W_U\otimes\mathbf{1}_{A^{q+1}}$. For even $q\ge 1$, Proposition~\ref{prop:LED-tau-independence} implies $\delta_E(Y)=0$ for all $Y\in C_q(\rho_{\mathrm{GHZ}})$, so LED$(q)$ cannot occur.
\end{proof}

\begin{table}[h]
\centering
\begin{tabular}{c|c|c}
degree $q$ & $\delta_E$ on $C_q(\rho_U)$ & operational consequence \\ \hline
$q=0$ & $\delta_E\equiv 0$ &
LED$(0)$ reduces to a usual separability witness on $U$ \\ \hline
even $q$ & $\delta_E(Y)=0~\forall Y\in C_q(\rho_U)$ &
the signed stabilize--reset pairing vanishes \\ \hline
odd $q$ & $\delta_E(Y)=\rho_U\otimes\tau_{\mathrm{aux}}^{\otimes(q+1)}~\forall Y\in C_q(\rho_U)$ &
LED$(q)$ holds iff $\rho_U$ is entangled
\end{tabular}
\caption{Parity collapse of the signed stabilize--reset pairing for the reindexed ancilla differential
(Proposition~\ref{prop:LED-tau-independence}). The conclusion is independent of the faithful choice of $\tau_{\mathrm{aux}}$ (Lemma~\ref{lem:tau-independence-cosimplicial}).}
\end{table}

\medskip
\noindent
These discussions are summarized in the following table.
\begin{center}
\renewcommand{\arraystretch}{1.15}
\begin{tabular}{l|p{0.3\linewidth}|p{0.35\linewidth}}
\hline
invariant & interpretation when it is $0$ & interpretation when it is $\neq 0$\\
\hline
$R^0=\ker(j)$ & If a global class exists, it is unique (injectivity of $j$). &
Multiple inequivalent global classes share the same marginals. Their differences are detected by witnesses.\\
\hline
Feasibility (PSD / sep.) & A global (PSD / separable) realization exists. &
Infeasibility is certified by a finite family of (separable) witnesses.\\
\hline
$E^q$ & No cohomological $(q\!+\!1)$-partite content confined to a single patch. Operational detection at degree $q$ may still be possible via the signed reset pairing.
& Some patch contains cohomological $(q+1)$-partite content. The first nonzero
cohomology degree is $q$.\\
\hline
\end{tabular}
\end{center}

If a chosen cover is too coarse to include the entangled block inside a single patch, the corresponding $E^{q}(\mathcal U)$ may look trivial simply because no patch detects the block. Refining the cover to add such a patch turns the appropriate $E^{q}$ on. In the colimit $\overline{E}^{\,q}=\varinjlim_{\mathcal U}E^{q}(\mathcal U)$ this refinement is automatic. On the single–patch cover $\{U\}$, $j$ is bijective, so $Q^{0}(\{U\})=R^{0}(\{U\})=0$. See Appendix B for procedure of entanglement test.

\section{A differential geometric perspective of the obstruction}
\label{sec:diff-form-obstruction}
Throughout this section, we fix a finite-dimensional complex Hilbert space $H\simeq\mathbb C^r$. We write $GL(H)$ for the set of all complex-linear isomorphisms on \(H\) and \(U(H):=\{U\in GL(H): U^\dagger U=I\}\simeq U(r)\). We consider a smooth manifold of parameters \(X\) and let \(\rho:X\to\mathcal D_{\mathrm{full}}(H)\) be a smooth field of full–rank density matrices on a fixed finite–dimensional Hilbert space \(H=\bigotimes_{j\in I}H_j\). The ordered presheaf \(\mathcal V\) remains as in the discrete theory, but we now organize the geometry on \(X\) via the principal \(U(r)\)–bundle of amplitudes, where \(r=\dim H\).

For each \(x\in X\) choose an amplitude \(W(x)\in GL(H)\) with \(\rho(x)=W(x)W(x)^{\dagger}\). The right action \(W\mapsto W u\) \((u\in U(H))\) leaves \(\rho\) unchanged, so local choices \(\{W_i\}\) over a good cover \(\mathcal U=\{U_i\}\) differ by unitaries \(u_{ij}:U_{ij}\to U(H)\),
\(
   W_i=W_j\,u_{ij}
\),
forming a Čech \(1\)–cocycle \(u=\{u_{ij}\}\in Z^{1}(\mathcal U,\,U(H))\).
The Uhlmann connection \(\mathcal A\in\Omega^{1}(X,\mathfrak u(H))\) is the unique unitary connection on the amplitude bundle whose horizontal spaces satisfy \(W^{\dagger}dW=dW^{\dagger}W\) (equivalently, the parallel transport maximizes Uhlmann fidelity) \cite{UHLMANN1986229}. Its curvature \( F=d A+ A\wedge A\in\Omega^{2}(X,\mathfrak u(H))\) is globally defined. The associated Chern–Weil forms are
\[
  c_k^{\mathrm{Uhl}}(\rho)\;:=\;\frac{1}{(2\pi i)^k}\,\Tr\big( F^{\,k}\big)
  \;\in\;\Omega^{2k}(X),\qquad d\,c_k^{\mathrm{Uhl}}(\rho)=0,
\]
giving de Rham classes \([c_k^{\mathrm{Uhl}}(\rho)]\in H^{2k}_{\mathrm dR}(X)\). For pure states this reduces to the usual Berry connection on the projective line bundle.

Given a smooth field of witnesses \(W:X\to C_{\sep}^{\ast}(H)\), the natural pairing with \(V\) produces scalar forms \(\omega^{(k)}_{W}:=\Tr(W F^{\,k})\), so that pairing the ordered Čech classes with \(W\) lands in de Rham cohomology. On overlaps \(U_{ij}\) the unitary transition maps \(u_{ij}\) define \(u\in Z^{1}(\mathcal U,U(H))\). For \(k\ge0\), the higher coboundaries \(\delta^{k}(u)\in C^{k+1}(\mathcal U,U(H))\) represent obstruction classes in Čech cohomology. Resetting order gives the standard class \([\delta^{k}(u)]\in H^{k+1}(\mathcal U,U(H))\). If we retain the order, we regard \(\delta^{k}(u)\) inside the ordered presheaf and then pair with smooth witness fields to compare with differential forms.

Since $C_{\mathrm{sep}}^{*}(H)$ depends on a fixed factorization $H=\bigotimes_j H_j$ and is not preserved by arbitrary $U(H)$, we treat a witness field as an $\mathrm{End}(H)$–valued weight (a section of the adjoint bundle) and make the Čech–de Rham comparison in a fixed gauge (or restrict the gauge group to local unitaries $\prod_j U(d_j)$).

\begin{HLblock}
\begin{theorem}
\label{thm:global-uhl}
Let $W_{\mathrm{sep}}:X\to C^*_{\mathrm{sep}}(H)$ be a smooth witness field which is covariantly constant with respect to the Uhlmann connection $A$ on the amplitude bundle, i.e.\ $D_AW_{\mathrm{sep}}=0$. Then for each $k\ge 1$ the Bott--Shulman--Stasheff/Dupont construction \cite{BOTT197643,DUPONT1976233} produces a Čech--de Rham cocycle $\Xi_k(W_{\mathrm{sep}},A,u)\in\mathrm{Tot}^{2k}$ whose de Rham class is
\begin{equation}
\label{eq:cocycle}
[\Xi_k(W_{\mathrm{sep}},A,u)]_{\mathrm{Tot}}
\;=\;
\Big[\frac{1}{(2\pi i)^k}\Tr\big(W_{\mathrm{sep}}F^{\,k}\big)\Big]_{\mathrm{dR}}
\in H^{2k}_{\mathrm{dR}}(X).
\end{equation}
In particular, for $W_{\mathrm{sep}}=\mathbf 1$ this recovers $[c_k^{\mathrm{Uhl}}(\rho)]_{\mathrm{dR}}$.
\end{theorem}
\end{HLblock}

\begin{proof}
Fix a good cover $\mathcal U=\{U_i\}$ of $X$. Choose local amplitudes $W_i:U_i\to GL(\mathcal H)$ with $\rho=W_iW_i^\dagger$. On overlaps $U_{ij}:=U_i\cap U_j$ they are related by unitary transition maps $u_{ij}:U_{ij}\to U(\mathcal H)$ via
\[
W_j = W_i\,u_{ij},
\]
so that on triple overlaps $U_{ijk}$ one has $u_{ij}u_{jk}=u_{ik}$.

\smallskip
\noindent
\emph{Step 1.}
We regard $W$ as a section of the adjoint bundle $\mathrm{Ad}(H)=X\times^{U(H)}\End(H)$, i.e.\ on overlaps $W$ transforms as $W|_{U_i}=u_{ij}^{-1}\,W|_{U_j}\,u_{ij}$. For each $i$ let $A_i\in\Omega^{1}(U_i,\mathfrak u(H))$ be the local Uhlmann connection form, with curvature $F_i=dA_i+A_i\!\wedge\!A_i$. Then on overlaps $A_j=u_{ij}^{-1}A_i\,u_{ij}+u_{ij}^{-1}du_{ij}$ and $F_j=u_{ij}^{-1}F_i\,u_{ij}$, while $W|_{U_j}=u_{ij}^{-1}W|_{U_i}u_{ij}$. Consequently the multilinear polynomial
\[
P_W(X_1,\dots,X_k):=\frac{1}{(2\pi i)^k}\Tr(WX_1\cdots X_k)
\]
is gauge-covariant in the following sense: for any unitary $u$,
\[
P_{u^{-1}Wu}\bigl(u^{-1}X_1u,\dots,u^{-1}X_ku\bigr)=P_W(X_1,\dots,X_k).
\]
In particular, if $u$ lies in the stabilizer $G_{W_0}$ (so that $u^{-1}W_0u=W_0$), then $P_{W_0}$ restricts to an $\mathrm{Ad}$-invariant polynomial on $\mathfrak g_{W_0}$.

\smallskip
\noindent
\emph{Step 2.}
Define on each $U_i$
\[
   \alpha_i \;:=\; P_{W}(F_i,\dots,F_i)
   \;=\;\frac{1}{(2\pi i)^k}\Tr\big(W\,F_i^{\,k}\big)
   \;\in\;\Omega^{2k}(U_i).
\]
The transformation laws imply $\alpha_i|_{U_{ij}}=\alpha_j|_{U_{ij}}$, so the $\alpha_i$ glue to a global form $\alpha$. Using $D_AW=0$ and the Bianchi identity $D_AF=0$, one has $d\alpha=0$.

\medskip
\noindent{Step 3.}
Fix a good cover $\mathcal U=\{U_i\}$ of $X$. Let
\[
C^{p,q}(\mathcal U):=\prod_{i_0<\cdots<i_p}\Omega^q(U_{i_0\cdots i_p}),
\qquad U_{i_0\cdots i_p}:=U_{i_0}\cap\cdots\cap U_{i_p}.
\]
The Čech coboundary is
\[
(\delta\omega)_{i_0\cdots i_{p+1}}
=\sum_{m=0}^{p+1}(-1)^m\,
\omega_{i_0\cdots\widehat{i_m}\cdots i_{p+1}}\big|_{U_{i_0\cdots i_{p+1}}}.
\]
We use the standard total differential
\[
D := \delta + (-1)^p d
\]
on $\mathrm{Tot}^n(\mathcal U)=\bigoplus_{p+q=n} C^{p,q}(\mathcal U)$.

\medskip
\noindent{Step 4.}
Fix a basepoint $x_0\in X$ and set $W_0:=W_{\mathrm{sep}}(x_0)\in \End(H)$.
Let
\[
G_{W_0}:=\{\,g\in U(H)\mid gW_0g^{-1}=W_0\,\}
\]
be the stabilizer (equivalently, the unitary centralizer) of $W_0$.

Since $D_AW_{\mathrm{sep}}=0$, parallel transport preserves $W_{\mathrm{sep}}$; in particular, the holonomy of $A$ lies in $G_{W_0}$. Equivalently, after choosing a good cover $\mathcal U=\{U_i\}$ and local gauges on each contractible $U_i$ so that $W_i\equiv W_0$ is constant, the transition functions satisfy $u_{ij}:U_{ij}\to G_{W_0}$.

For this reduced structure group $G_{W_0}$, the symmetric multilinear functional
\[
P_{W_0}(X_1,\dots,X_k):=\frac{1}{(2\pi i)^k}\Tr\bigl(W_0\,X_1\cdots X_k\bigr)
\]
is $\Ad_{G_{W_0}}$--invariant (because $W_0$ is central in $G_{W_0}$). Therefore the Bott--Shulman--Stasheff/Dupont construction \cite{BOTT197643,DUPONT1976233} applies to the $G_{W_0}$--bundle and produces, on a good cover, a canonical \v{C}ech--de~Rham cocycle
\(
\Xi_k(W_{\mathrm{sep}},A,u)\in \mathrm{Tot}^{2k}(\mathcal U)
\)
whose $(0,2k)$--component is
\[
\omega^{(0)}=(\omega_i^{(0)}),\qquad
\omega_i^{(0)}=\frac{1}{(2\pi i)^k}\Tr\bigl(W_i\,F_i^k\bigr),
\]
and whose top \v{C}ech component is the corresponding group cocycle $\Phi_{W_{\mathrm{sep}}}(u)$.

\smallskip
\noindent
\emph{Step 5.}
If $D_AW=0$, then the $2k$-forms $\omega^{(0)}_i$ agree on overlaps and glue to a global closed form
$\omega^{(0)}=\frac{1}{(2\pi i)^k}\Tr(WF^{\,k})$ on $X$. Hence the total cocycle is $D$-cohomologous to
$(\omega^{(0)},0,\dots,0)$ in $\mathrm{Tot}^{2k}$.

\end{proof}

For $\rho(x)=|\psi(x)\rangle\langle\psi(x)|$, the amplitude bundle reduces to the tautological line bundle, $A$ becomes the Berry connection, and Theorem~\ref{thm:global-uhl} recovers the familiar identification of Berry–Chern classes.

Fix a good cover $\mathcal U=\{U_i\}$ and local trivialisations with amplitudes $W_i$ and connection $A_i$ (curvatures $F_i$) and unitaries $u_{ij}$. Let $W=(W_i)_i$ be any smooth adjoint–valued witness field with $W_i=\operatorname{Ad}_{u_{ij}^{-1}}(W_j)$ on overlaps.
Consider the Bott–Shulman descent for the invariant polynomial $(X_1,\dots,X_k)\mapsto \operatorname{Tr}(W_i X_1\cdots X_k)$, which produces forms $\omega^{(p)}\in C^{p,\,2k-p}$ satisfying $\delta\omega^{(p)}=d\omega^{(p+1)}$ for $p=0,\dots,k-1$ and $\delta\omega^{(k)}=(-1)^k\,\Phi_W(u)$, where
\[
  \Phi_W(u)_{i_0\cdots i_k}
  \;=\;\frac{1}{(2\pi i)^k}\operatorname{Tr}\!\big(
     W_{i_0}\,\theta_{i_0i_1}\wedge\theta_{i_1i_2}\wedge\cdots\wedge\theta_{i_{k-1}i_k}
  \big),\qquad \theta_{ij}:=u_{ij}^{-1}du_{ij}.
\]
With the total differential $D=d+(-1)^p\delta$ on $C^{p,\bullet}$, one has
\[
D\Omega=\Phi_W(u)\in C^{k+1,k},
\qquad \Omega:=\omega^{(0)}\oplus\cdots\oplus\omega^{(k)}\in\bigoplus_{p=0}^k C^{p,\,2k-p}.
\]
Thus, in the total complex, the Čech class $\Phi_W(u)$ is cohomologous to $\omega^{(0)}=\frac{1}{(2\pi i)^k}\operatorname{Tr}(W F^k)$. Under the Čech–de Rham isomorphism for a good cover,
\[
\big\langle W,\ \delta^k(u)\big\rangle\;=\;\big[\Phi_W(u)\big]_{\check C}\;=\;
\Big[\tfrac{1}{(2\pi i)^k}\operatorname{Tr}(W F^k)\Big]_{\mathrm{dR}}.
\]
If, moreover, $D_A W=0$, then $d\,\omega^{(0)}=0$, and we may replace the descent data by a $D$–cohomologous representative with $\omega^{(p)}\equiv 0$ for all $p\ge1$. Hence the class is represented by the global closed form $\tfrac{1}{(2\pi i)^k}\operatorname{Tr}(W F^k)$.

\section{\label{sec:QGL}Quantum entanglement index and quantum geometric Langlands correspondence}
\subsection{Motivations}
Our formulation would give a natural ``quantum extension''  of the Atiyah–Singer index theorem \cite{26aef1a0-6391-389c-92d4-cfc9f0b3e4c6} from the perspective of quantum information geometry (see Remark \ref{remark:index}). While traditionally it has been used in various settings of physics and gauge theories \cite{Fukaya:2019qlf,Yamashita:2020nkf}, the conventional formulation does not quantify quantum entanglement in quantum many-body systems. However, a number of examples suggests that the index and quantum physics has significant relations (e.g., in nuclear physics \cite{PhysRevD.78.074033}, high energy theory \cite{Alvarez-Gaume:1983zxc} and condensed matter \cite{graf2013bulk}).

\subsection{\label{sec:index}Definition and basic properties}
Let $X$ be a smooth, closed, even–dimensional spin manifold with chiral Dirac operator $D_X^+:\Gamma(\mathscr S^+)\to\Gamma(\mathscr S^-)$, and let $\widehat A(TX)$ denote the $\widehat A$–class. Let $E\to X$ be a Hermitian complex vector bundle with a unitary connection $A$ and curvature $F_A=dA+A\wedge A$.

We consider a smooth witness field $W\in\Gamma(\End(E))$. For each $x\in X$, write the spectral decomposition $W(x)=\sum_j \lambda_j(x)\,P_j(x)$ into real eigenvalues $\lambda_j(x)$ with mutually orthogonal projections $P_j(x)$. Using the (Borel) functional calculus, define the spectral projectors
\[
P_+(x):=\sum_{\lambda_j(x)>0}P_j(x),\qquad
P_-(x):=\sum_{\lambda_j(x)<0}P_j(x),\qquad
P_0(x):=\sum_{\lambda_j(x)=0}P_j(x).
\]
Set $S:=\operatorname{sgn}(W):=P_+-P_-$ and $|S|:=P_++P_-=I-P_0$. Then we have:
\[
P_\pm=\frac{|S|\pm S}{2},\qquad P_0=I-|S|.
\]
If $W$ is invertible, $P_\pm=(I\pm S)/2$.

Assume now that $W$ is $A$–parallel, i.e. $D_AW=0$, where $D_A$ is the covariant derivative. By the resolvent/Riesz functional calculus, differentiating under the contour integral for $P_\pm$ gives $D_AP_\pm=0$ (and hence $D_AP_0=0$). Thus the ranks of $P_\bullet$ are constant and
\[
E \;=\; E^+ \oplus E^0 \oplus E^-,
\qquad
E^\bullet:=\operatorname{im}P_\bullet,
\]
is an $A$–parallel splitting. Applying $D_A^2$ to $P_\bullet$ yields
$[F_A,P_\bullet]=D_A^2P_\bullet=0$, so the curvature preserves the splitting and block–diagonalizes:
\[
F_A \;=\; F^+ \oplus F^0 \oplus F^-,\qquad F^\bullet:=F_A\big|_{E^\bullet}.
\]
Since $P_\bullet$ commute with $F_A$ and are orthogonal idempotents, one has the weighted trace identity
\begin{equation}
\label{eq:identity}
\Tr\bigl(S\,e^{F_A/2\pi i}\bigr) =\Tr\bigl(e^{F^+/2\pi i}\bigr)-\Tr\bigl(e^{F^-/2\pi i}\bigr) =\mathrm{ch}(E^+)-\mathrm{ch}(E^-).
\end{equation}
Then we introduce the following.
\begin{HLblock}
\begin{definition}
\label{def:quantum-index}
Define the \emph{quantum entanglement index} (QEI) as:
\begin{equation}
\label{eq:index}
\Ind_S(D_X\!\otimes E)\;:=\;\ind(D_X\!\otimes E^+)\;-\;\ind(D_X\!\otimes E^-)\;.
\end{equation}
\end{definition}
\end{HLblock}
\noindent
Here each twisted chiral Dirac operator $D_X^\pm\!\otimes\mathbf 1_{E^\pm}$ is elliptic, hence Fredholm on Sobolev spaces. Therefore its kernel and cokernel are finite dimensional and the analytic index is an integer:
\[
\ind\bigl(D_X\!\otimes E^\pm\bigr)\in\mathbb Z \quad\Longrightarrow\quad \Ind_S(D_X\!\otimes E)\in\mathbb Z.
\]

By the Atiyah–Singer index theorem for twisted Dirac operators, for any Hermitian bundle $F\to X$ with unitary connection,
\[
\ind(D_X\!\otimes F)\;=\;\big\langle \widehat A(TX)\wedge \ch(F),[X]\big\rangle.
\]
Applying this to $F=E^\pm$ and subtracting gives the cohomological formula
\[
\Ind_S(D_X\!\otimes E)\;=\;\Big\langle \widehat A(TX)\wedge\bigl(\mathrm{ch}(E^+)-\mathrm{ch}(E^-)\bigr),[X]\Big\rangle,
\]
and, by the weighted trace identity above,
\[
\Ind_S(D_X\!\otimes E)\;=\;\Big\langle \widehat A(TX)\wedge \Tr\bigl(S\,e^{F_A/2\pi i}\bigr),[X]\Big\rangle.
\]
Since the left–hand side is an analytic index, the differential form expression on the right is an integer. In summary, we have the following.
\begin{HLblock}
\begin{corollary}
\label{cor:quantum-index}
The index \eqref{eq:index} is an integer: 
\[
\hlEq{
\Ind_S(D_X\!\otimes E) \;=\; \Big\langle \widehat A(TX)\wedge \Tr\bigl(Se^{F_A/2\pi i}\bigr),[X]\Big\rangle\in\mathbb Z.
}
\]
\end{corollary}
\end{HLblock}

An important application of our index is mixed‑state geometry, while pure states are commonly considered in physics in this context. 

\begin{remark}\label{remark:index}
The index \eqref{eq:index} can be interpreted as the quantum version of the Atiyah–Singer index. The reasons are as follows. If $W$ detects quantum entanglement, $S$ selects a parallel subbundle $E^+\oplus E^-$ that encodes those ``entanglement‑active'' fiber directions. The spectral decomposition $W=W_+-W_-$ with supports $E^\pm$ gives 
\[
\Tr(W\rho)=\Tr(W_+\rho)-\Tr(W_-\rho)\;,
\]
so negativity ($\Tr(W\rho)<0$) can only occur if the state family places sufficient weight on the negative spectral sector $E^-$, i.e. the only sector that can contribute negatively to the witness pairings. Therefore the index \eqref{eq:index} measures a global, entanglement-induced imbalance that cannot be generated by separable (classical) families. 
\end{remark}

\subsection{Geometric Langlands correspondence with quantum entanglement}
Our discussions can be naturally related to the geometric Langlands program. The goal of this subsection is to link our theory of quantum many-body systems with established results (see \cite{Frenkel:2005pa}, for example). 

Let $(E,A)$ be a principal $U(r)$–bundle over a smooth base $X$ with unitary connection $A$ and curvature $F_A$. As before, let $W\in\Gamma(\End(E))$ be a smooth witness field with $D_AW=0$, write $S=\sgn(W)$, and denote the $A$–parallel spectral splitting by $E=E^+\oplus E^0\oplus E^-$. It is convenient to view parallel endomorphisms via the adjoint connection. On $\End(E)$, the induced connection is $\nabla^{\ad}:=D_A$ with curvature $(\nabla^{\ad})^2=\ad(F_A)$. Hence the centralizer subbundle
\[
\mathfrak z_A:=\ker(\ad(F_A))=\{Z\in\End(E):[F_A,Z]=0\}\subset\End(E)
\]
is preserved by $\nabla^{\ad}$ and the restriction is flat. The sheaf of horizontal sections of $(\mathfrak z_A,\nabla^{\ad})$ is the local system of endomorphisms that are parallel for $A$. Fixing a basepoint $x_0\in X$ and identifying $E_{x_0}\cong\mathbb C^r$, the space of global horizontal sections identifies with the commutant of the holonomy: 
\[
\Gamma_{\nabla^{\ad}}(X,\mathfrak z_A)\cong \End_{\Hol(A)}(\mathbb C^r)=\{T\in\End(\mathbb C^r):\Ad(h)(T)=T\ \forall\,h\in\Hol(A)\}.
\]
In particular, there exists a non‑scalar $A$–parallel Hermitian endomorphism $W$ (equivalently, a nontrivial $A$–parallel splitting $E=E^+\oplus E^-$ with $S=\sgn(W)$) if and only if the holonomy representation is reducible, i.e. $\Hol(A)$ preserves a proper subspace of $\mathbb C^r$. In other words, $\End_{\Hol(A)}(\mathbb C^r)$ contains a non‑scalar Hermitian. With the reduction of structure group to a Levi subgroup $U(r)\to U(r_+)\times U(r_-)$ determined by $S$, the curvature block‑diagonalizes and one has \eqref{eq:identity}.

Now let $C$ be a smooth complex projective curve, and consider holomorphic vector bundles with unitary Chern connections ($c_1=\frac{1}{2\pi i}\int_C\Tr F_A$). On a curve $\widehat A(TC)=1$, therefore the QEI reduces to
\[
\Ind_S\big|_C=\frac{1}{2\pi i}\int_C\Tr(SF_A)=\deg(E^+)-\deg(E^-).
\]
At a point $p\in C$, a positive elementary modification is a short exact sequence
\[
0\longrightarrow E\ \xrightarrow{\ \iota\ }\ E'\ \longrightarrow\ \mathbb C_p\longrightarrow 0,
\]
which increases the degree by $1$, i.e. $\deg E'=\deg E+1$. A Hecke modification of a split bundle $E$ at $p$ acts by an elementary modification on $E^+$ and/or $E^-$ at $p$ while leaving the other Levi factor(s) unchanged. More generally, given a coweight \(\lambda=(m_1,\dots,m_{r_+};\,n_1,\dots,n_{r_-})\in\Z^{r_+}\times\Z^{r_-}\), one performs \(|m_i|\) (resp. \(|n_j|\)) elementary modifications on the $i$‑th line in $E^+$ (resp. the $j$‑th line in $E^-$), with the sign of $m_i$ and $n_j$ determining positive/negative type.

These correspondences have a simple effect on the QEI along $C$. A positive elementary modification of $E^+$ at $p$ changes the index by $+1$, while a positive elementary modification of $E^-$ changes it by $-1$. For a general coweight $\lambda$, the total jump is
\[
\Delta\,\Ind_S\big|_C=\sum_{i=1}^{r_+} m_i\ -\ \sum_{j=1}^{r_-} n_j\ =:\ \langle S,\lambda\rangle.
\]
If one performs the Hecke modifications at marked points $\{p_a\}$ with coweights $\{\lambda_a\}$, the net change is the sum of the local signed charges:
\[
\Delta\,\Ind_S\big|_C=\sum_a \langle S,\lambda_a\rangle.
\]

Now let us consider a smooth base \(X\) and the amplitude bundle \((E,A)\) of a smooth full-rank quantum state family \(\rho\). For a given $\mathcal{D}$-module of witness fields, we can naturally ask its quantum geometric Langlands (QGL) correspondence for quantum many-body systems.

\begin{definition}
\label{def:QGL}
A \emph{QGL spectral datum} is \((E,A,W)\) with $D_AW=0$.
It determines \(L=\mathrm{GL}_{r_+}\times\mathrm{GL}_{r_-}\subset \mathrm{GL}_r\), and the graded Chern character
\(\text{ch}(E^+)-\text{ch}(E^-)=\Tr(Se^{F_A/2\pi i})\).
\end{definition}

\begin{conjecture}\label{conj:qgl-eis}
Let \((E,A,W)\) be a QGL spectral datum on \(C\), and assume \((E,\nabla)\) is flat (after complexification). Writing \(S=\sgn(W)\) and \(L=\mathrm{GL}_{r_+}\times\mathrm{GL}_{r_-}\), the \(\mathrm{GL}_r\)–eigenobject attached to \((E,\nabla)\) on \(\mathrm{Bun}_{\mathrm{GL}_r}\) is Eisenstein‑induced from \(\mathrm{Bun}_L\). Consequently, Hecke eigenvalues factor through \(L\), giving \(\Ind_S=\Ind(\cdot,E^+)-\Ind(\cdot,E^-)\).
\end{conjecture}

\begin{conjecture}\label{conj:qgl-hecke}
In the setting above, for a Hecke modification of coweight \(\lambda\) at \(p\in C\), the automorphic object changes by the Eisenstein functor of type \(\lambda\), and the QEI satisfies
\(\Delta\Ind_S=\langle S,\lambda\rangle\).
\end{conjecture}

These conjectures can be resolved by applying the results of \cite{gaitsgory2024proof,arinkin2024proof,campbell2024proof,arinkin2024proof4}.

\vskip0.3cm
Physics interpretations of the statements and conjectures are as follows. A quantum phase transition (QPT) in the witness–selected sector is precisely the event where the \(A\)-parallel splitting \(E=E^+\oplus E^-\) changes by a Hecke modification, and this is detected by a quantized jump of the invariant associated to quantum entanglement. Away from such a locus the virtual class \([E^+]-[E^-]\) is constant, hence the QEI \eqref{eq:index} is locally constant, so it can change only when a Hecke modification alters the Levi factors (i.e. a genuine sector change), which is the physical signature of a QPT. On a two–parameter submanifold \(C\), this reduces to the quantum number \(\nu_{\mathrm{ent}}=\tfrac{1}{2\pi i}\!\int_C\!\Tr(SF)\), and a unit Hecke correspondence in the \(E^+\) block produces the quantized jump \(\Delta\nu_{\mathrm{ent}}=+1\) (and similarly \(-1\) for \(E^-\)). Following the standard argument of phase transitions, this jump would accompany a gap closing in the \(S\)-graded sector (typically visible as a spike in the fidelity/BKM metric and often in entanglement entropy (see also Section \ref{sec:phys})). Thus, integer jumps of \(\nu_{\mathrm{ent}}\) serve as a robust, topological diagnosis of QPTs as Hecke modifications in parameter space.

\vskip0.3cm
Fix a QGL spectral datum $(E,A,W)$ and define the $L$–character as
\[
\chi_S:\; L\longrightarrow U(1),\qquad \chi_S(g_+,g_-)\;=\;\det(g_+)\,\det(g_-)^{-1}.
\]
Then for every smooth closed curve $\gamma\subset X$ the \emph{entanglement Wilson loop} is the evaluation of this character on the monodromy of the $L$–local system determined by $(E,A)$:
\begin{equation}\label{eq:W-loop-as-character}
\mathfrak{W}_S(\gamma)
\;=\;\exp\!\Big(-\oint_\gamma \Tr(SA)\Big)
\;=\;\chi_S\big(\Hol_E(\gamma)\big)
\;=\;\frac{\det\!\big(\Hol_{E_+}(\gamma)\big)}{\det\!\big(\Hol_{E_-}(\gamma)\big)}.
\end{equation}

Under geometric Langlands, $(E,A)$ is the spectral input and automorphic Hecke functors act on $\mathcal D$–modules on $\mathrm{Bun}_{GL_r}$. The entanglement grading supplied by $S$ (via the QEI) is expected to cause Hecke eigenvalues to factor through the Levi $L$ determined by $S$. Concretely, the loop $\gamma$ around a point $x\in X$ picks out the eigenvalue
\[
\chi_S\big(\Hol_E(\gamma_x)\big)\;=\;\mathfrak{W}_S(\gamma_x),
\]
so $\mathfrak{W}_S$ is the $1$–dimensional character by which the Hecke kernel acts on the $S$–graded sector. Equivalently, the eigenvalue of $H_{x,V}$ on the automorphic object attached to $(E,A,W)$ factors through the $L$–character $\chi_S$.

\smallskip

A Hecke modification of coweight $\lambda$ at $x$ changes the $L$–eigenvalue by the signed charge $\langle S,\lambda\rangle$. At the level of \eqref{eq:W-loop-as-character} this appears as
\[
\mathfrak{W}_S(\gamma_x)\;\longmapsto\;e^{2\pi i\langle S,\lambda\rangle}\,\mathfrak{W}_S(\gamma_x),
\]
which matches the jump $\Delta\Ind_S=\langle S,\lambda\rangle$ of the QEI and, on oriented surfaces, the jump $\Delta\nu_{\mathrm{ent}}$ of the entanglement‑induced number \eqref{eq:EE-invariant}.

\smallskip
For the Satake/Hecke description it is convenient to replace $U(1)$ by $\mathbb G_m$ and view $\chi_S^{\mathrm{alg}}:\,L_\mathbb C\to\mathbb G_m$. Through spherical Satake, $\chi_S^{\mathrm{alg}}$ corresponds to a one–dimensional representation of the dual Levi ${}^{L}L\subset{}^{L}G$, and the associated spherical kernel acts by the scalar $\chi_S(\Hol_E(\gamma_x))=\mathfrak{W}_S(\gamma_x)$ on the automorphic side. In the absence of entanglement, the duality between a Wilson loop and a Hecke operator is expected \cite{Kapustin:2006pk,Frenkel:2005pa}. The arguments here are applied to a topological phase of matter in \cite{2026arXiv260113467I}.

\section{\label{sec:phys}Outlook: Implications to quantum physics}
This section offers a brief outlook and sketches future directions.
Detailed numerical implementations and case studies for concrete quantum many-body models
are presented in the companion paper~\cite{2026arXiv260113467I}. Here we focus on the conceptual and
mathematical framework.
\subsection{Toward practical detection of entanglement via QEI in quantum many-body systems}
In this work, on 2-manifolds $X$, for a given 2-form $F$, we define an entanglement curvature 2‑form
\[
\Omega^{(W)}=\Tr(WF)\;,
\]
which can be used to analyze the topological space of entangled states. With $S:=\operatorname{sgn}(W)=P_+-P_-$, then we obtain the entanglement-induced number 
\begin{equation}
\label{eq:EE-invariant}
\nu_\text{ent}:=\frac{1}{2\pi i}\int_X\Tr(SF)\;,
\end{equation}
which is analogous to the TKNN/Berry number \cite{PhysRevLett.49.405}, but now filtered by the chosen entanglement structure\footnote{In conventional topological insulators, the quantization of topological charge is rooted in the anomaly cancellation between the bulk and the edges/surfaces, rather than originating primarily from quantum entanglement. For example, the quantum Hall effect can be described by a single-particle theory, making quantum entanglement irrelevant in this context.}. (Note that when $D_AW=0$, $\Omega^{(W)}$ is a closed 2‑form $d\Omega^{(W)}=0$ due to the Bianchi identity $D_AF=0$.) This gives a way to project out purely classical/locally separable contributions, leaving an obstruction that is operationally tied to quantum entanglement. 

Across a topology‑changing transition, $\nu_\text{ent}$ jumps, which is a signature of phase transition induced solely by quantum entanglement, corresponding to the Hecke modification. 

We also highlight the implications for condensed matter physics and nuclear physics, suggesting that the mathematical quantities proposed in this article may be implemented and discovered experimentally. Physically, $\nu_\text{ent}$ could be observed as an \emph{entanglement-induced (topological) effect}, where the conductance receives a correction due to entanglement. Furthermore, the QEI has a natural connection to chiral physics: it measures the entanglement between left- and right-handed modes. The traditional Atiyah-Singer index counts the imbalance, $n_L - n_R = \Tr[\rho J_5] =\tfrac{1}{2\pi i} \int \Tr[F]$, between the number of left- and right-moving modes. Here $J_5$ is the axial current, and the number of each mode is given by the index of the corresponding Dirac operator: $n_\bullet=\ind D^\bullet$. In contrast, the QEI detects their entanglement: $ \Tr[\rho S J_5] = \tfrac{1}{2\pi}\int \Tr[S F]$. Both entanglement-induced phenomena can be detected by two-dimensional condensed matter systems, such as the Haldane systems and (multilayer) graphenes. For nuclear physics, it would correspond to measuring the entanglement between the left- and right-handed movers in the chiral magnetic effect (CME) \cite{PhysRevD.78.074033,Kharzeev:2022ydx}.

The following is a prescription for numerical simulations. The Uhlmann curvature on parameter space admits the Bogoliubov–Kubo–Mori inner product \cite{doi:10.1143/JPSJ.12.570,10.1143/PTP.33.423,Petz:1999xrh}
\begin{equation}
\label{eq:uhlmann-kubo}
   F_{\lambda\sigma}\;=\; \frac{i}{2}\int_{0}^{1}\!ds\;\Tr\big(\rho^{s}\,[\partial_\lambda K,\partial_\sigma K]\,\rho^{1-s}\big), \qquad K:=-\ln\rho.
\end{equation}
For numerics, one may use a Fukui–Hatsugai–Suzuki–type discretization on a $N_x\times N_y$ mesh (torus $T^2$) \cite{Fukui:2005wr}. Let $\Phi(k)$ denote Uhlmann amplitudes in parallel-transport gauge at the plaquette nodes $k=(i,j)$, and set link variables
\[
U_\mu(k)\;=\;\frac{\langle \Phi(k),\,\Phi(k+\hat\mu)\rangle}{|\langle \Phi(k),\,\Phi(k+\hat\mu)\rangle|}\,,
\qquad \mu\in\{x,y\}.
\]
Then the plaquette curvature is
\[
F_{ij}\;=\;U_x(i,j)\,U_y(i+1,j)\,U_x(i,j+1)^{-1}\,U_y(i,j)^{-1},
\]
and $\nu_{\mathrm{ent}}$ is obtained by the normalized lattice sum of $\,\arg F_{ij}\,$ with a witness.

\subsection{Relation to high energy theory}
In general covariant settings, metric variations insert the stress tensor. Consider a case where $X$ is a product $X=\Sigma\times\Sigma$ of Riemann surfaces \((\Sigma,g)\). Writing
\(\langle\cdot\rangle_\rho^{c}\) for connected correlators, we have
\begin{equation}\label{eq:stress-kubo}
   F_{\lambda\sigma}
   \;=\; -\!\int_{\Sigma}\!\! d^dx\,\sqrt{g(x)}
          \int_{\Sigma}\!\! d^dy\,\sqrt{g(y)}\
          \big\langle \tfrac12 T_{\mu\nu}(x),\tfrac12 T_{\rho\kappa}(y)\big\rangle^{c}_{\rho}
          \partial_{[\lambda} g^{\mu\nu}(x)\partial_{\sigma]} g^{\rho\kappa}(y).
\end{equation}
Pairing \eqref{eq:stress-kubo} with a smooth witness field \(W\) (as in Section~\ref{sec:diff-form-obstruction}) yields scalar \(2\)-forms on \(\mathcal M\) representing de Rham classes after the Čech–de Rham identification.

When \(\rho[g]\) tracks the geometry, the Einstein equation is a standard backreaction model:
\begin{equation}
\label{eq:einstein-semi}
  G_{\mu\nu}[g]\;+\;\Lambda\,g_{\mu\nu}\;=\;8\pi G\,\big\langle T_{\mu\nu}\big\rangle_{\rho[g]}\,.
\end{equation}
Linearizing around a background solution \(g_0\) along \(g(\lambda)\) gives
\begin{equation}
\label{eq:lin-einstein}
  \partial_\lambda\big\langle T_{\mu\nu}\big\rangle_{\rho[g]}
  \;=\;\frac{1}{8\pi G}\Big(\partial_\lambda G_{\mu\nu}[g]
  +\Lambda\,\partial_\lambda g_{\mu\nu}\Big)
  \;+\;\text{(local anomaly/contact terms)}\,,
\end{equation}
so variations of the metric are traded for variations of geometric tensors. Combining \eqref{eq:stress-kubo} and \eqref{eq:lin-einstein}, we find that there exists a
bidistribution \(\mathcal K^{\alpha\beta\gamma\delta}_{W}(x,y,g_0)\), depending on the state and the witness, such that for geometry directions \(\lambda,\sigma\),
\begin{equation}
\label{eq:F-Einstein}
  \big\langle W,F_{\lambda\sigma}\big\rangle
  \;=\;\iint_{\Sigma\times\Sigma}\!\!\mathcal K^{\alpha\beta\gamma\delta}_{W}(x,y,g_0)\;
     \partial_{[\lambda}\!\big(G_{\alpha\beta}+\Lambda g_{\alpha\beta}\big)(x)\;
     \partial_{\sigma]}\!\big(G_{\gamma\delta}+\Lambda g_{\gamma\delta}\big)(y)\,
     d\Sigma_x\,d\Sigma_y.
\end{equation}
Thus, after pairing with witnesses, the Uhlmann curvature along metric directions is a quadratic functional of the linearized Einstein tensor. In regimes where the correlation length is short relative to curvature scales, \(\mathcal K_{W}\) localizes, and \eqref{eq:F-Einstein} reduces to a local curvature density built from \(R_{\mu\nu\rho\sigma}\) (and its contractions), in line with Chern–Weil locality.

\smallskip
We now return to a closed oriented four‑dimensional smooth $X$. The following follows essentially immediately from standard index theory:
\begin{equation}
\label{eq:4D-full}
\Ind_S(D_X\!\otimes E) =\int_X\!\left\{\frac12\big(c_1^2(E^+)-c_1^2(E^-)\big)-\big(c_2(E^+)-c_2(E^-)\big)\ -\ \frac{r_+-r_-}{24}\,p_1(TX)\right\}.
\end{equation}
We then consider its reduction to a 2d case. For 2d CFT on a compact Riemann surface \((\Sigma,g)\),
\[
\frac{1}{2\pi i}\,\mathrm{Tr}\,F \;=\; -\,\frac{c}{24\pi}\,R^{(2)}\,(i\,g_{z\bar z})\,dz\wedge d\bar z,
\]
so the Berry/Uhlmann curvature density is proportional to the Gaussian curvature, with proportionality given by the central charge \(c\). This matches the local form of \eqref{eq:F-Einstein} and reflects the 2d trace anomaly structure. Let $S=\mathrm{sgn}(W)$ be $A$–parallel with parallel splitting $E=E^{+}\oplus E^{0}\oplus E^{-}$ of ranks $r_{\pm}=\mathrm{rank}\,E_{\pm}$ and $r=\mathrm{rank}\,E$. In vacuum 2d CFT the curvature is central, $F=f\,\mathbf{1}_{r}$, hence
\[
\mathrm{Tr}(S F)=(\mathrm{Tr}\,S)\,f=(r_{+}-r_{-})\,f
\quad\text{and}\quad\frac{1}{2\pi i}\int_{\Sigma}\mathrm{Tr}(S F)= -\,\frac{r_{+}-r_{-}}{r}\,\frac{c}{6}\,\chi(\Sigma).
\]
Equivalently, when $F$ is central one has 
\[
c_{1}(E_{+})-c_{1}(E_{-}) \;=\; \frac{r_{+}-r_{-}}{r}\,c_{1}(E),
\]
so the entanglement-induced 2d invariant is just a rank‑weighted multiple of the unfiltered one. This is the 2d analogue of the $(r_{+}-r_{-})$ gravitational coefficient appearing in eq.~\eqref{eq:4D-full}. If we consider a dynamical system, this would correspond to measuring the entanglement in the flow within a curved spacetime background. 

For any subregion \(A\) with modular Hamiltonian \(K_A[g]=-\ln\rho_A[g]\), the first law of entanglement yields \cite{Casini:2011kv}
\begin{equation}
\label{eq:first-law}
  \delta S_A\;=\;\delta\langle K_A\rangle\;=\;\int_A \xi^\mu\,\delta\big\langle T_{\mu\nu}\big\rangle\, d\Sigma^\nu,
\end{equation}
with \(\xi^\mu\) the modular flow vector (e.g.\ a Killing/boost field in symmetric
setups). Using \eqref{eq:lin-einstein}, variations of \(S_A\) are therefore driven
by variations of the Einstein tensor, paralleling the curvature response
\eqref{eq:F-Einstein}.

The entanglement obstruction on parameter space, calculated from the unitary overlaps of local amplitudes, matches the de Rham classes of the Chern–Weil forms of the Uhlmann connection when paired with suitable test functions. Hence, when the state depends smoothly on background data, the visible topological entanglement obstruction reduces to ordinary characteristic classes of a unitary connection. Physically, its curvature is encoded by stress–tensor two–point functions as in \eqref{eq:stress-kubo}.


\clearpage
\section*{Appendix A. Summary of Notations}
\begin{center}
\renewcommand{\arraystretch}{1.2}
\begin{tabular}{@{}p{0.20\textwidth}p{0.75\textwidth}@{}}
\hline
Symbol & Meaning \\
\hline
$I=\{1,\dots,N\}$ & Site index set; $U\subset I$ a finite subsystem/patch. \\
$H_U=\bigotimes_{j\in U} H_j$ & Hilbert space on $U$. \\
$\mathcal V(U)=\mathrm{Herm}(H_U)$ & Real vector space of Hermitian operators on $H_U$. \\
$\mathcal D(U)$, $\mathcal D_{\mathrm{sep}}(U)$ & Density operators on $U$; separable density operators. \\
$C_{\mathrm{sep}}(U)$ & Separable cone in $\mathcal V(U)$ (closed, convex, pointed). \\
$C^*_{\mathrm{sep}}(U)$ & Dual cone (entanglement witnesses). \\
$\mathrm{Wit}(U)$, $W(U)$ & Witness cone and its real linear span $W(U)=\mathrm{span}_{\mathbb{R}}\mathrm{Wit}(U)$. \\
$r^U_{V}$ & Restriction on witnesses: $r^U_{V}(W)=\mathrm{Tr}_{U\setminus V}\!\big[W(\mathbf 1_V\otimes \tau_{U\setminus V})\big]$. Depending on the context, it is also used for the partial trace of states. \\
$\delta_C$ & Čech coboundary. \\
$\delta_E$ & The ancilla–cosimplicial differential associated to $E(Z)=\mathrm{Tr}_{\mathrm{aux}}(Z)\otimes\tau_{\mathrm{aux}}$. \\
$Q^0(\mathcal U)$, $R^0(\mathcal U)$ & Cokernel/kernel of restriction $j:\mathcal V(I)\to H^0(\mathcal U,\mathcal V)$ (gluing/uniqueness defects). \\
$E^q(\mathcal U)$ & Local entanglement groups defined via $H^q(C^{0,\bullet}(\mathcal U),\delta_E)$. \\
$\overline{Q}^0,\overline{R}^0,\overline{E}^q$ & Refinement colimits over covers. \\
\hline
\end{tabular}
\end{center}

\section*{Appendix B. Flowchart of entanglement test}
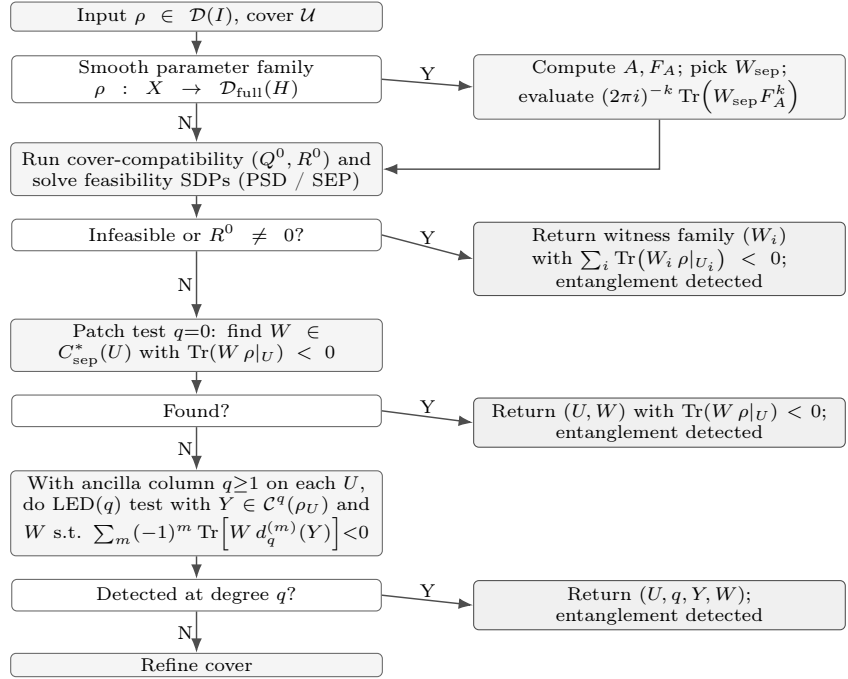
\begin{figure}[H]
\centering
\begin{tikzpicture}[
  >=Latex,
  every node/.style={font=\scriptsize},
  RowSep/.store in=\RowSep, RowSep=3.0mm,
  ColSep/.store in=\ColSep, ColSep=12mm,
  BlockW/.store in=\BlockW, BlockW=48mm,
  block/.style  ={rectangle, rounded corners=2pt, draw=black!60, fill=gray!08,
                  align=center, inner sep=2pt, text width=\BlockW},
  dec/.style    ={rectangle, rounded corners=2pt, draw=black!60, fill=white,
                  align=center, inner sep=1.6pt, text width=\BlockW, minimum height=4.2mm},
  result/.style ={rectangle, rounded corners=2pt, draw=black!70, fill=black!6,
                  align=center, inner sep=2pt, text width=\BlockW},
  flow/.style   ={-Latex, line width=0.55pt, draw=black!70},
  lab/.style    ={font=\scriptsize, inner sep=1pt}
]
\matrix[matrix of nodes, row sep=\RowSep, column sep=\ColSep, nodes in empty cells] (M) {
  \node[block] (start) {Input $\rho\in\mathcal D(I)$, cover $\mathcal U$}; & \\
  \node[dec]   (fam)   {Smooth parameter \mbox{family} $\rho:X\to\mathcal D_{\mathrm{full}}(H)$};                                      & \node[block]  (geo)    {Compute $A,F_A$; pick $W_{\rm sep}$; evaluate $(2\pi i)^{-k}\Tr(W_{\rm sep}F_A^k)$}; \\
  \node[block] (cech)  {Run cover‑compatibility $(Q^0,R^0)$ and solve feasibility SDPs (PSD / SEP)}; & \\
  \node[dec]   (feas)  {Infeasible or $R^0\neq 0$?};                             & \node[result] (witfam) {Return witness family $(W_i)$ with $\sum_i\Tr(W_i\,\rho|_{U_i})<0$; entanglement detected}; \\
  \node[block] (qzero) {Patch test $q{=}0$: find $W\in C^*_{\rm sep}(U)$ with $\Tr(W\,\rho|_U)<0$}; & \\
  \node[dec]   (qzok)  {Found?};                                                 & \node[result] (q0out)  {Return $(U,W)$ with $\Tr(W\,\rho|_U)<0$; entanglement detected}; \\
  \node[block] (anc)   {With ancilla column $q{\ge}1$ on each $U$, do LED$(q)$ test with $Y\in\mathcal C^q(\rho_U)$ and $W$ s.t. $\sum_m(-1)^m\Tr[W\,d_q^{(m)}(Y)]{<}0$}; & \\
  \node[dec]   (anck)  {Detected at degree $q$?};                                & \node[result] (ancout) {Return $(U,q,Y,W)$; \\entanglement detected}; \\
  \node[block] (ref)   {Refine cover};                                           & \\
};

\draw[flow] (start) -- (fam);
\draw[flow] (fam.east) -- node[lab, above]{Y} (geo.west);
\draw[flow] (fam) -- node[lab, left]{N} (cech);
\draw[flow] (geo.south) |- (cech);

\draw[flow] (cech) -- (feas);
\draw[flow] (feas.east) -- node[lab, above]{Y} (witfam.west);
\draw[flow] (feas) -- node[lab, left]{N} (qzero);

\draw[flow] (qzero) -- (qzok);
\draw[flow] (qzok.east) -- node[lab, above]{Y} (q0out.west);
\draw[flow] (qzok) -- node[lab, left]{N} (anc);

\draw[flow] (anc) -- (anck);
\draw[flow] (anck.east) -- node[lab, above]{Y} (ancout.west);
\draw[flow] (anck) -- node[lab, left]{N} (ref);
\end{tikzpicture}
\end{figure}

\begin{remark}
The flowchart should be read as a hierarchy of sufficient tests.  For a fixed cover $\mathcal U$ and bounded ancilla order $q$, failure to detect entanglement at the last stage does not imply separability of $\rho$.  In principle, every entangled state admits an entanglement witness, so completeness can only be expected in a formal limit where one allows either (i) patches as large as the full system $U=I$ (reducing to a global witness test), and/or (ii) unbounded witness/ancilla size. Both limits are generically intractable (the separability problem is NP-hard).
\end{remark}

\section*{Conflict of interest}
The author declares no competing interests

\section*{Data availability}
The manuscript has no associated data.

\bibliographystyle{amsalpha}
\bibliography{ref}

\newcommand{\etalchar}[1]{$^{#1}$}
\providecommand{\bysame}{\leavevmode\hbox to3em{\hrulefill}\thinspace}
\providecommand{\MR}{\relax\ifhmode\unskip\space\fi MR }
\providecommand{\MRhref}[2]{%
  \href{http://www.ams.org/mathscinet-getitem?mr=#1}{#2}
}
\providecommand{\href}[2]{#2}
\begin{thebibliography}{TKNdN82}

\bibitem[AB11]{abramsky2011sheaf}
Samson Abramsky and Adam Brandenburger, \emph{The sheaf-theoretic structure of non-locality and contextuality}, New Journal of Physics \textbf{13} (2011), no.~11, 113036.

\bibitem[ABC{\etalchar{+}}24a]{arinkin2024proof}
Dima Arinkin, Dario Beraldo, J~Campbell, L~Chen, J~Faergeman, Dennis Gaitsgory, K~Lin, S~Raskin, and N~Rozenblyum, \emph{Proof of the geometric langlands conjecture ii: Kac-moody localization and the fle}, arXiv preprint arXiv:2405.03648 (2024).

\bibitem[ABC{\etalchar{+}}24b]{arinkin2024proof4}
Dima Arinkin, Dario Beraldo, Lin Chen, Joakim Faergeman, Dennis Gaitsgory, Kevin Lin, Sam Raskin, and Nick Rozenblyum, \emph{Proof of the geometric langlands conjecture iv: ambidexterity}, arXiv preprint arXiv:2409.08670 (2024).

\bibitem[AG83]{Alvarez-Gaume:1983zxc}
Luis Alvarez-Gaume, \emph{{Supersymmetry and the Atiyah-Singer Index Theorem}}, Commun. Math. Phys. \textbf{90} (1983), 161.

\bibitem[And79]{33675660-3119-37a8-a5d2-fa5a40dfb227}
Joel Anderson, \emph{Extensions, restrictions, and representations of states on $c^\ast$- algebras}, Transactions of the American Mathematical Society \textbf{249} (1979), no.~2, 303--329.

\bibitem[AS68]{26aef1a0-6391-389c-92d4-cfc9f0b3e4c6}
M.~F. Atiyah and I.~M. Singer, \emph{The index of elliptic operators: I}, Annals of Mathematics \textbf{87} (1968), no.~3, 484--530.

\bibitem[BS00]{berenstein2000coadjoint}
Arkady Berenstein and Reyer Sjamaar, \emph{Coadjoint orbits, moment polytopes, and the hilbert-mumford criterion}, Journal of the American Mathematical Society \textbf{13} (2000), no.~2, 433--466.

\bibitem[BSS76]{BOTT197643}
R~Bott, H~Shulman, and J~Stasheff, \emph{On the de rham theory of certain classifying spaces}, Advances in Mathematics \textbf{20} (1976), no.~1, 43--56.

\bibitem[BT13]{BottTu}
R.~Bott and L.W. Tu, \emph{Differential forms in algebraic topology}, Graduate Texts in Mathematics, Springer New York, 2013.

\bibitem[CCF{\etalchar{+}}24]{campbell2024proof}
Justin Campbell, Lin Chen, Joakim Faergeman, Dennis Gaitsgory, Kevin Lin, Sam Raskin, and Nick Rozenblyum, \emph{Proof of the geometric langlands conjecture iii: compatibility with parabolic induction}, arXiv preprint arXiv:2409.07051 (2024).

\bibitem[CDKW14]{2014CMaPh.332....1C}
Matthias {Christandl}, Brent {Doran}, Stavros {Kousidis}, and Michael {Walter}, \emph{{Eigenvalue Distributions of Reduced Density Matrices}}, Communications in Mathematical Physics \textbf{332} (2014), no.~1, 1--52.

\bibitem[CHM11]{Casini:2011kv}
Horacio Casini, Marina Huerta, and Robert~C. Myers, \emph{{Towards a derivation of holographic entanglement entropy}}, JHEP \textbf{05} (2011), 036.

\bibitem[CM06]{Christandl2006}
Matthias Christandl and Graeme Mitchison, \emph{The spectra of quantum states and the kronecker coefficients of the symmetric group}, Communications in Mathematical Physics \textbf{261} (2006), no.~3, 789--797.

\bibitem[CM23]{Collins:2021xyq}
Beno{\^\i}t Collins and Colin McSwiggen, \emph{{Projections of orbital measures and quantum marginal problems}}, Trans. Am. Math. Soc. \textbf{376} (2023), no.~08, 5601--5640.

\bibitem[Dup76]{DUPONT1976233}
Johan~L. Dupont, \emph{Simplicial de rham cohomology and characteristic classes of flat bundles}, Topology \textbf{15} (1976), no.~3, 233--245.

\bibitem[FFM{\etalchar{+}}20]{Fukaya:2019qlf}
Hidenori Fukaya, Mikio Furuta, Shinichiroh Matsuo, Tetsuya Onogi, Satoshi Yamaguchi, and Mayuko Yamashita, \emph{{The Atiyah{\textendash}Patodi{\textendash}Singer Index and Domain-Wall Fermion Dirac Operators}}, Commun. Math. Phys. \textbf{380} (2020), no.~3, 1295--1311.

\bibitem[FHS05]{Fukui:2005wr}
Takahiro Fukui, Yasuhiro Hatsugai, and Hiroshi Suzuki, \emph{{Chern numbers in a discretized Brillouin zone: Efficient method to compute (spin) Hall conductances}}, J. Phys. Soc. Jap. \textbf{74} (2005), 1674--1677.

\bibitem[FKW08]{PhysRevD.78.074033}
Kenji Fukushima, Dmitri~E. Kharzeev, and Harmen~J. Warringa, \emph{Chiral magnetic effect}, Phys. Rev. D \textbf{78} (2008), 074033.

\bibitem[Fre07]{Frenkel:2005pa}
Edward Frenkel, \emph{{Lectures on the Langlands program and conformal field theory}}, {Les Houches School of Physics: Frontiers in Number Theory, Physics and Geometry}, 2007, pp.~387--533.

\bibitem[GP13]{graf2013bulk}
Gian~Michele Graf and Marcello Porta, \emph{Bulk-edge correspondence for two-dimensional topological insulators}, Communications in Mathematical Physics \textbf{324} (2013), no.~3, 851--895.

\bibitem[GR24]{gaitsgory2024proof}
Dennis Gaitsgory and Sam Raskin, \emph{Proof of the geometric langlands conjecture i: construction of the functor}, arXiv preprint arXiv:2405.03599 (2024).

\bibitem[GS82]{guillemin1982convexity}
Victor Guillemin and Shlomo Sternberg, \emph{Convexity properties of the moment mapping}, Inventiones mathematicae \textbf{67} (1982), no.~3, 491--513.

\bibitem[GT09]{Guhne:2008qic}
Otfried G{\"u}hne and G{\'e}za T{\'o}th, \emph{{Entanglement detection}}, Phys. Rept. \textbf{474} (2009), 1--75.

\bibitem[HHH96]{Horodecki:1996nc}
Michal Horodecki, Pawel Horodecki, and Ryszard Horodecki, \emph{{On the necessary and sufficient conditions for separability of mixed quantum states}}, Phys. Lett. A \textbf{223} (1996), 1.

\bibitem[IR26]{2026arXiv260113467I}
Kazuki {Ikeda} and Steven {Rayan}, \emph{{Quantum Entanglement, Stratified Spaces, and Topological Matter: Towards an Entanglement-Sensitive Langlands Correspondence}}, arXiv e-prints (2026), arXiv:2601.13467.

\bibitem[Kha22]{Kharzeev:2022ydx}
Dmitri~E. Kharzeev, \emph{{Chiral magnetic effect in heavy ion collisions and beyond}}, 4 2022.

\bibitem[Kir84]{kirwan1984convexity}
Frances Kirwan, \emph{Convexity properties of the moment mapping, iii}, Inventiones mathematicae \textbf{77} (1984), no.~3, 547--552.

\bibitem[KS59]{25f8c26d-c81c-39c1-9128-a71146ac84ff}
Richard~V. Kadison and I.~M. Singer, \emph{Extensions of pure states}, American Journal of Mathematics \textbf{81} (1959), no.~2, 383--400.

\bibitem[Kub57]{doi:10.1143/JPSJ.12.570}
Ryogo Kubo, \emph{Statistical-mechanical theory of irreversible processes. i. general theory and simple applications to magnetic and conduction problems}, Journal of the Physical Society of Japan \textbf{12} (1957), no.~6, 570--586.

\bibitem[KW07]{Kapustin:2006pk}
Anton Kapustin and Edward Witten, \emph{{Electric-Magnetic Duality And The Geometric Langlands Program}}, Commun. Num. Theor. Phys. \textbf{1} (2007), 1--236.

\bibitem[Mor65]{10.1143/PTP.33.423}
Hazime Mori, \emph{Transport, collective motion, and brownian motion}, Progress of Theoretical Physics \textbf{33} (1965), no.~3, 423--455.

\bibitem[MSS15]{7f2d66a9-2d0a-3ab6-8756-1cda100696a6}
Adam~W. Marcus, Daniel~A. Spielman, and Nikhil Srivastava, \emph{Interlacing families ii: Mixed characteristic polynomials and the kadison-singer problem}, Annals of Mathematics \textbf{182} (2015), no.~1, 327--350.

\bibitem[Pet99]{Petz:1999xrh}
D{\'e}nes Petz, \emph{{Monotone metrics on matrix spaces}}, Linear Algebra Appl. \textbf{244} (1999), 81--96.

\bibitem[Ter02]{2001quant.ph..1032T}
Barbara~M Terhal, \emph{Detecting quantum entanglement}, Theoretical computer science \textbf{287} (2002), no.~1, 313--335.

\bibitem[TKNdN82]{PhysRevLett.49.405}
D.~J. Thouless, M.~Kohmoto, M.~P. Nightingale, and M.~den Nijs, \emph{Quantized hall conductance in a two-dimensional periodic potential}, Phys. Rev. Lett. \textbf{49} (1982), 405--408.

\bibitem[Uhl86]{UHLMANN1986229}
Armin Uhlmann, \emph{Parallel transport and ``quantum holonomy'' along density operators}, Reports on Mathematical Physics \textbf{24} (1986), no.~2, 229--240.

\bibitem[Ume62]{umegaki1962conditional}
Hisaharu Umegaki, \emph{{Conditional expectation in an operator algebra. IV. Entropy and information}}, Kodai Math. Sem. Rep. \textbf{14} (1962), no.~2, 59 -- 85.

\bibitem[Yam21]{Yamashita:2020nkf}
Mayuko Yamashita, \emph{{A Lattice Version of the Atiyah{\textendash}Singer Index Theorem}}, Commun. Math. Phys. \textbf{385} (2021), no.~1, 495--520.

\end{thebibliography}
\end{document}